\documentclass[a4paper, 10 pt]{article}
\usepackage[utf8x]{inputenc}
\usepackage[margin=1in]{geometry}
\usepackage{amsmath}
\usepackage{amssymb}
\usepackage{amsthm}
\usepackage{textpos}
\usepackage{bbold}
\usepackage{mathrsfs}
\usepackage{hyperref}
\usepackage{url}
\usepackage{textcase}
\usepackage{lmodern}
\usepackage{pdflscape}
\usepackage[toc,page]{appendix}
\usepackage{authblk}
\usepackage{mathrsfs}
\usepackage[square,sort,comma,numbers]{natbib}
\usepackage{accents}
\newcommand\vv[1]{\underaccent{\vec}{#1}}
\usepackage{graphicx}% Include figure files
\usepackage{dcolumn}% Align table columns on decimal point
\usepackage{bm}% bold math
\newcommand{\ketbra}[2]{|{#1}\rangle \langle {#2}|}
\newcommand{\is}[4]{\langle #1 \,  | \, #2 \rangle \langle #3 \, | \, #4  \rangle}
\newcommand{\pk}[3]{\langle #1 \,  | \, #2 \, | #3  \rangle}
\newcommand{\ir}[3]{\langle #1 \,  | \, #2 \rangle \langle #2 \, | \, #3  \rangle}
\newcommand{\ketbrat}[4]{| \widetilde{#1}_{#2} \rangle \langle \widetilde{#3}_{#4}|}
\newcommand{\bra}[1]{\langle {#1}|}
\newcommand{\ket}[1]{| {#1} \rangle}
\newcommand{\tket}[2]{\ket{{\widetilde{#1}}_{#2}}}
\newcommand{\braket}[2]{\langle #1 | #2 \rangle}
\newcommand{\tbraket}[4]{\langle \widetilde{#1}_{#2} | \widetilde{#3}_{#4} \rangle}
\newcommand{\brat}[2]{ \langle \widetilde{#1}_{#2}|}

\newtheorem{mydef}{Definition}
\newtheorem{theorem}{Theorem}[section]

\begin{document}
\date{}
\title{Minimum Error Discrimination for an Ensemble of Linearly Independent Pure States}
\author[1]{Tanmay Singal\thanks{stanmay@imsc.res.in}}
\author[1]{Sibasish Ghosh\thanks{sibasish@imsc.res.in}}
\affil[1]{Optics and Quantum Information Group,The Institute of Mathematical Sciences, CIT Campus, Taramani, Chennai, 600 113, India}
\providecommand{\keywords}[1]{\textbf{\textit{Keywords:---}} #1}

\maketitle

\begin{abstract}
Inspired by the work done by Belavkin [\url{Belavkin V. P., Stochastics, 1, 315 (1975)}], and independently by Mochon, [\url{Phys. Rev. A 73, 032328, (2006)}], we formulate the problem of minimum error discrimination of any ensemble of $n$ linearly independent pure states by stripping the problem of its rotational covariance and retaining only the rotationally invariant aspect of the problem. This is done by embedding the optimal conditions in a matrix equality as well as matrix inequality. Employing the implicit function theorem in these conditions we get a set of first-order coupled ordinary non-linear differential equations which can be used to drag the solution from an initial point (where solution is known) to another point (whose solution is sought). This way of obtaining the solution can be done through a simple Taylor series expansion and analytic continuation when required. Thus, we \emph{complete} the work done by Belavkin and Mochon by ultimately leading their theory to a solution for the minimum 
error 
discrimination problem of linearly independent pure state ensembles. We also compare the computational complexity of our technique with a barrier-type interior point method of SDP and show that our technique is computationally as efficient as (actually, a bit more than) the SDP algorithm, with the added advantage of being much simpler to implement.
\end{abstract}


\section{\label{intro}Introduction}
In the class of quantum state discrimination problems minimum error discrimination (MED) is one of the oldest. The problem arises because nonorthogonal states aren't perfectly distinguishable. Thus any measurement aimed at distinguishing among states cannot hope to do so without some error. Different measurement strategies have different performance strength (measured in terms of the average probability of success). Given that the states cannot be distinguished perfectly there must be some measurement criterion which gives the maximum probability of success. To find what this measurement strategy is, is the problem of MED.

The setting in MED is as follows: Alice has a fixed ensemble of states $\{ \rho_1, \rho_2, \cdots, \rho_n\}$ where $\rho_i$ are positive semi-definite operators of trace $1$ acting on some Hilbert space $\mathcal{H}$ of dimension $n$. She selects one of these states ($\rho_i$, say) with probability $p_i  \in \{p_1, p_2, \cdots, p_n\}$ ($p_i\; >\;0,\; \sum_{i=1}^n \;p_i\; =\;1$, $p_i$'s are refered to as apriori probabilities) and gives it to Bob. Bob knows that Alice has selected the state from the set $\{\rho_i\}_{i=1}^{n}$ with apriori probabilities $p_i$ and his job is to figure out which state he has been given using an $n$-element POVM. In MED, Bob's measurement strategy is constrained in the following way: there is a one-to-one correspondence between elements in Alice's ensemble $\{ p_i, \rho_i \}_{i=1}^{n}$ and Bob's POVM elements $\{ E_i\}_{i=1}^{n}$ (where $E_i \ge 0$, and $\sum_{i=1}^{n} E_i = \mathbb{1}_n$, which is the identity operator acting on $\mathcal{H}$) so that when the $i$-th measurement 
outcome clicks, Bob infers Alice gave him the $i$-th state 
from her ensemble. Since $\rho_1$, $\rho_2$, $\cdots$, $\rho_n$ don't generally lie on orthogonal supports, errors are likely to occur. Bob's job is to find the optimal POVM for minimizing the average probability of this error or equivalently maximizing the average probability of success. 

There are other variants to the quantum state discrimination problem \cite{Per}, \cite{Croke}, \cite{Walgate}. The most popular among them is called unambiguous state discrimination, in which, if one has to perform state discrimination over an ensemble of $n$ states $\{ p_i, \rho_i \}_{i=1}^{n}$, the measurement strategy used has $n+1$ outcomes, where, just as in the MED case, there is a one-to-one correspondence between ensemble elements $\rho_i$ and the $i$-th POVM element $E_i$. Furthermore, the POVM must be constrained so that when Alice sends Bob the $i$-th state, the $j$-th POVM element won't click where $ j \neq i, n+1$. The trade-off is that Bob can say nothing about the state which Alice gave him when the $(n+1)$-th POVM element clicks. Heuristically one expects that one cannot discriminate unambiguously among a set of linearly dependent states; this was proven true later \cite{Chef}.  

Coming back to MED, necessary and sufficient conditions for the optimal POVM for any ensemble were given by Holevo \cite{Hol} and Yuen et al. \cite{Yuen} independently. Yuen et al. cast MED into a convex optimization problem for which numerical solutions are given in polynomial time\footnote{That is, polynomial in $dim \mathcal{H}$.}. While there are quite a number of numerical techniques to obtain the optimal POVM \cite{Eldar, Jezek, Hel2, Tyson1}, for very few ensembles has the MED problem been solved analytically. Some of these include an ensemble of two states \cite{Hel}, ensembles whose density matrix is maximally mixed state \cite{Yuen}, equiprobable ensembles that lie on the orbit of a unitary, \cite{Ban,Chou,BS}, and mixed qubit states \cite{Bae, Ha}. In \cite{Kwon}, many interesting properties of the MED problem have been elucidated using geometry of $N$ qudit states. An upper bound for the optimal success probability was given in \cite{Bae1}. Comparing the existing results for MED with 
that of the unambiguous state discrimination problem, it is seen that the latter has been solved for more kinds of ensembles ensembles than the former \cite{Per, Chef,Pang, Bergou, Ray, Herzog,Janos, Som}. 

\textbf{Linearly Independent Pure State Ensembles:} For an ensemble of $n$  linearly independent pure states ($n$-LIP), given by $\widetilde{P} \equiv \{p_i, \ketbra{\psi_i}{\psi_i} \}_{i=1}^{n}$ (where $\ket{\psi_i}$'s span $\mathcal{H}$), certain properties which the optimal POVM should satisfy has been given in the literature on MED already:

\begin{itemize}
 \item [(i)] The optimal POVM is a unique rank one projective measurement \cite{Ken,Hel,Carlos}.
 \item [(ii)] The optimal POVM for MED of $\widetilde{P}$ is the pretty good measurement (PGM) of another ensemble, $\widetilde{Q}\equiv \left\{ q_i>0, \ketbra{\psi_i}{\psi_i} \right\}_{i=1}^{n}$\footnote{While (i) is subsumed by (ii) (as the PGM of the ensemble $\widetilde{Q}$ is a rank-one projective measurement), it is beneficial to, separately, emphasize it.} \cite{Bela,Mas,Carlos}. Note that the $i$-th state in $\widetilde{P}$ and $\widetilde{Q}$ are the same for all $1 \le i \le n$, whereas the probabilities are generally not. Additionally, in \cite{Carlos}, it is explicitly shown that the ensembles $ \widetilde{P}, \widetilde{Q}$ are related through an invertible map. 
\end{itemize}

To formalize this invertible relation between $\widetilde{P}$ and $\widetilde{Q}$ we will now introduce a few definitions.

\begin{mydef}
\label{DefEns}
$\mathcal{E}$ is the set of all ensembles comprising of $n$ LI pure states. Hence, any ensemble in $\mathcal{E}$ is of the form  $\widetilde{P} = \{ p_i>0, \ketbra{\psi_i}{\psi_i} \}_{i=1}^{n}$ where $\ket{\psi_1},\ket{\psi_2},\cdots,\ket{\psi_n}$ are LI.\end{mydef}

$\mathcal{E}$ is a $(2n^2-n-1)$ real parameter set.

\begin{mydef}
\label{DefPro}
$\mathcal{P}$ is the set of all rank one projective measurements on the states of $\mathcal{H}$; an element in $\mathcal{P}$ is of the form $\{ \ketbra{v_i}{v_i} \}_{i=1}^{n}$ where $\braket{v_i}{v_j}=\delta_{ij},\; \forall \; 1 \leq i,j \leq n$.
\end{mydef}

$\mathcal{P}$ is an $n(n-1)$ real parameter set. From point (i) above we see that the optimal POVM for $\widetilde{P} \in \mathcal{E}$ is a unique element in $\mathcal{P}$. Thus, one can define the \emph{optimal POVM map}, $\mathscr{P}$, in the following way: 
\begin{mydef}
 \label{mathscrP}
$\mathscr{P}: \mathcal{E} \longrightarrow \mathcal{P}$ is such that $ \mathscr{P}\left( \widetilde{P} \right)$ is \emph{the} optimal POVM for MED of $\widetilde{P}\in \mathcal{E}$. 
\end{mydef}

Let $PGM$ denote the PGM map, i.e., $PGM:\mathcal{E} \longrightarrow \mathcal{P}$ is such that $PGM\left( \widetilde{Q} \right)$ is the PGM of $\widetilde{Q} \in \mathcal{E}$, i.e. (refer to \cite{Carlos}), $PGM\left(\widetilde{Q}\right)=\left\{\rho_q^{-\frac{1}{2}}q_i \ketbra{\psi_i}{\psi_i}\rho_q^{-\frac{1}{2}} \right\}_{i=1}^{n}$, where $\rho_q = \sum_{i=1}^{n} q_i \ketbra{\psi_i}{\psi_i}$. 

Then (ii) above says that there exists an invertible map, which we label by $\mathscr{R}$, which can be defined in the following way:

\begin{mydef}
\label{mathRdef}
$\mathscr{R}: \mathcal{E} \longrightarrow \mathcal{E}$ is a bijection such that \begin{equation}  \label{mathRR}  \mathscr{P}\left( \widetilde{P} \right) =PGM \left( \mathscr{R} \left( \widetilde{P} \right) \right), \; \forall \; \widetilde{P} \in \mathcal{E}.\end{equation}

\end{mydef}

Knowing $\mathscr{R}$ would solve the problem of MED for LI pure state ensembles. While the existence of the invertible function $\mathscr{R}$ has been proven \cite{Bela,Carlos}, unfortunately, it isn't known - neither analytically nor computationally for arbitrary ensemble $\widetilde{P}$. Fortunately $\mathscr{R}^{-1}$ is known \cite{Mas, Bela, Carlos} i.e., having fixed the states $\{ \ket{\psi_i} \}_{i=1}^{n}$ one can give $p_i$ in terms of the $q_i$: let $G_q >0$ represent the gram matrix of the ensemble $\widetilde{Q}$, i.e., $\left( {G_q} \right)_{ij}=\sqrt{q_i q_j} \braket{\psi_i}{\psi_j}$, and let  ${G_q}^{\frac{1}{2}}$ represent the positive square root of $G_q$; let $G$ denote the gram matrix of $\widetilde{P}$, i.e., $G_{ij} = \sqrt{p_ip_j}\braket{\psi_i}{\psi_j}, \; \forall \; 1 \leq i, j \leq n$; then diagonal elements of $G$ can be written as functions of $q_i$ and matrix elements of $G_q^\frac{1}{2}$
\begin{equation}
\label{piqi}
G_{ii} = p_i = C \frac{q_i}{ \left( {G_q}^{\frac{1}{2}} \right)_{ii} } , \; \forall \; 1 \leq i \leq n,
\end{equation}
where $C$ is the normalization constant\footnote{$q_i >0, \; \forall \; 1 \leq i \leq n$. This comes from the definition of $\mathcal{E}$ and that $\{ q_i, \ketbra{\psi_i}{\psi_i} \}_{i=1}^{n} \in \mathcal{E}$. Also, $\left( G^\frac{1}{2} \right)_{ii} > 0, \; \forall \; 1 \leq i \leq n$. This is because $G^\frac{1}{2}$, being the positive square root of $G$ (gram matrix for the linearly independent vectors $\{ \sqrt{q_i} \ket{\psi_i} \}_{i=1}^{n}$) is positive definite and the diagonal elements of a positive definite matrix have to be greater than zero.}, $$C = \left(\sum_{j=1}^{n}  \dfrac{q_j}{\left( {G_q}^{\frac{1}{2}} \right)_{jj}} \right)^{-1}.$$ 

This tells us what $\mathscr{R}^{-1}$ is: $$\mathscr{R}^{-1}\left( \left\{ q_i, \ketbra{\psi_i}{\psi_i} \right\}_{i=1}^{n} \right) =\left\{ p_i, \ketbra{\psi_i}{\psi_i} \right\}_{i=1}^{n},$$ where $p_i$ and $q_i$ are related by equation \eqref{piqi}.

It is more convenient to define $\mathscr{R}^{-1}$ and $\mathscr{R}$ on the set of gram matrices, which we will denote by $\mathcal{G}$. 
\begin{mydef}
\label{defG}
$\mathcal{G}$ is the set of all positive definite $n \times n$ matrices of trace one. 
\end{mydef}

Note\footnote{Associating each $G \in \mathcal{G}$ with an $n \times n$ density matrix of rank $n$, we see that $\mathcal{G}$ is the same as the interior of the generalized Bloch sphere for $n$ dimensional systems. Hence $\mathcal{G} \subset \mathbb{R}^{n^2-1}$.} that $\mathcal{G}$ is convex and is also open in $\mathbb{R}^{n^2-1}$.

Define $\mathscr{R}_\mathcal{G}^{-1}:\mathcal{G} \longrightarrow \mathcal{G}$ by $\mathscr{R}_\mathcal{G}^{-1}(G_q) = G$, using relation \eqref{piqi}. We know that $\mathscr{R}^{-1}$ is invertible on $\mathcal{E}$ (from \cite{Carlos}); this implies that $\mathscr{R}_\mathcal{G}^{-1}$ is invertible on $\mathcal{G}$, i.e., $\mathscr{R}_\mathcal{G}$ exists. Equation \eqref{piqi} tells us that $\mathscr{R}_\mathcal{G}^{-1}$ is continuous in $\mathcal{G}$. Since $\mathcal{G} \subset \mathbb{R}^{n^2-1}$ is open\footnote{The topology of $\mathcal{G}$ is that which is induced on it by the Hilbert-Schmidt norm. Note that this is equivalent to the Euclidean metric of $\mathbb{R}^{n^2-1}$}, the invariance of domain theorem \cite{Spivak1} tells us that $\mathscr{R}_\mathcal{G}^{-1}$ is a homeomorphism on $\mathcal{G}$. This means that $\mathscr{R}_\mathcal{G}$ is also continuous on $\mathcal{G}$. 

To be able to express what $\mathscr{R}$ is one needs to be able to solve the $n$ equations \eqref{piqi} for $q_i$ in terms of $p_j$'s and $\ket{\psi_j}$'s. These equations are too complicated for one to hope to solve: to begin with one doesn't even have an explicit closed form expression for $G^{\frac{1}{2}}$ in terms of the matrix elements of $G$ for arbitrary $n$. Even for the cases when $n = 3,4$, where one \emph{can} obtain such a closed form expression for $G^{\frac{1}{2}}$, the nature of the equations is too complicated to solve analytically. This tells us that it is hopeless to obtain $q_i$ as a closed form expression in terms of $\{ p_i, \ket{\psi_i} \}_{i=1}^{n}$. A similar sentiment was expressed earlier \cite{Tyson}. While a closed form expression of the solution seems too difficult to obtain (and even if obtained, too cumbersome to appreciate) giving an \emph{efficient technique} to compute $q_i$ from $\{ p_i, \ket{\psi_i} \}_{i=1}^{n}$ establishes that the relation given by equation \eqref{mathRR} 
along with technique (to compute $q_i$) provides a solution for MED of an ensemble of $n$-LIPs. 

To achieve such a technique we recast the MED problem for any ensemble $\widetilde{P}$ in terms of a matrix equation and a matrix inequality using the gram matrix $G$ of $\widetilde{P}$. The matrix equation and the inequality are equivalent to the optimality conditions that the optimal POVM has to satisfy, i.e., the optimal conditions given by Yuen et al \cite{Yuen}. Recasting the problem in this form gives us three distinct benefits.
\begin{itemize}                                                                                                                                                                                                                                                                                                                                                                                                                                                               \item[(1)] It helps us to \emph{explicitly} establish that the optimal POVM for $\widetilde{P}$ is given by the PGM 
of another ensemble of the form $\widetilde{Q}$ (i.e., relation in equation \eqref{mathRR} is made explicit in the matrix equality and matrix inequality conditions).
\item[(2)] MED is actually a rotationally invariant problem, i.e., the optimal POVM, $\{ E_i \}_{i=1}^{n}$, varies covariantly under a unitary transformation, $U$, of the states: $$\ketbra{\psi_i}{\psi_i} \rightarrow U\ketbra{\psi_i}{\psi_i}U^\dag \Longrightarrow E_i\rightarrow U E_i U^\dag. $$ This makes it desirable to subtract out the rotationally covariant aspect of the solution and, so, cast the problem in a rotationally invariant form. This is achieved through the aforesaid matrix equality and inequality.
\item[(3)] It gives us a technique to compute $q_i$.                   
\end{itemize}

For (3) we need to compute $\mathscr{R}_\mathcal{G}(G)$ for a given $G \in \mathcal{G}$. This is done by using the analytic implicit function theorem which tells us that $\mathscr{R}_\mathcal{G}$ is an analytic function on $\mathcal{G}$. Specifically, we will vary $G$ from one point in $\mathcal{G}$, at which we know what the action of $\mathscr{R}_\mathcal{G}$ is, to another point in $\mathcal{G}$, at which we want to know what the action of $\mathscr{R}_\mathcal{G}$ is. 

Further on, since our technique rests on the theory of the MED problem for $n$-LIP ensembles, it is expected that the algorithm our technique offers is computationally as efficient as or more efficient than existing techniques. We show that this is indeed the case, particularly by directly employing Newton's method to solve the matrix inequality. This adds to the utility of our technique.

The paper is divided into the following sections. In Section \ref{MEDP} we go into detail about what MED is and elaborate on the optimality conditions, and specify what they look like for $n$-LIPs. In Section \ref{SSGMM} we recast the MED problem for LI pure states in a rotationally invariant form. In Section \ref{empIFT} IFT is employed to solve the rotationally invariant conditions, which were developed in the previous section; in subsection \ref{algcompl} of section \ref{empIFT} the computational complexity of our algorithm is compared with a standard SDP technique. We conclude the paper in section \ref{conclusion}.

\section{The MED Problem: The Conditions of Optimality}

\label{MEDP}

When the states $\ket{\psi_i}$ are pairwise orthogonal, i.e., $\braket{\psi_i}{\psi_j} =0, \, \forall \, 1 \leq i,j \leq n$, one can perfectly distinguish among them by performing the projective measurement $\{\ketbra{\psi_i}{\psi_i}\}_{i=1}^{n}$. In general the states $\{\ket{\psi_i} \}_{i=1}^n$ aren't pairwise orthogonal and in such a case, it may happen that despite being given $\ket{\psi_i}$, one's measurement output is $j$, leading to an error. The average probability of such errors, $P_e$, is then given by
\begin{subequations}
\begin{equation}
\label{Pe}
P_e = \sum_{\substack{i,j=1 \\ i\neq j}}^n p_i \bra{\psi_i} E_j \ket{\psi_i},
\end{equation} and the average probability of success $P_s$ is given by
\begin{equation}
\label{Ps}
P_s =  \sum_{i=1}^n p_i \bra{\psi_i} E_i \ket{\psi_i},
\end{equation}
\end{subequations}
where $\{ E_j \}_{j=1}^{n}$ represents an $n$ element POVM ($n$-POVM). Note that $ P_s + P_e =1$. Our task is to maximize $P_s$ by choosing an appropriate POVM in the set of $n$-element POVMS. We refer to the maximum value of $P_s$ as $P_s^{{max}}$. The maximum success probability $P_s^{max}$ is given by
\begin{equation}
\label{Pmax}
\begin{split}
P_{s}^{max}& = \text{Max} \;  \{P_s | \; \{ E_j \}_{j=1}^{n} \text{is an $n$-POVM}\}.
\end{split}
\end{equation}
The set of $n$-POVMs is convex, i.e., if $\{E_i\}_{i=1}^{n}$ and $\{E'_i\}_{i=1}^{n}$ are $n$-POVMs, then so is $\{ p E_i + (1-p) E'_i \}_{i=1}^{n}$, $\forall \; 0 \leq p \leq 1$. Hence MED is a constrained convex optimization problem. To every such a constrained convex optimization problem (called the primal problem) there is a corresponding dual problem which provides a lower bound (if the primal problem is a constrained minimization) or an upper bound (if the primal problem is a constrained maximization) to the quantity being optimized (called the objective function). Under certain conditions these bounds are tight implying that one can obtain the solution for the primal problem from its dual. We then say that there is no duality gap between both problems. For MED, there is no duality gap between the primal and dual problems; thus the dual problem can be solved to obtain the optimal POVM \cite{Yuen}. The dual problem is given as follows \cite{Yuen}:
 
\begin{equation}
\label{dual}
\text{Min} \:  \text{Tr}(Z), \; \text{subject to: }  Z \geq p_i \ketbra{\psi_i}{\psi_i}\, , \; \forall \; 1 \leq i \leq n.
\end{equation}
Then $P_s^{max}$ is given by $P_s^{max} = \text{ Min } Tr(Z)$.

Also the optimal $n$-POVM, $\{E_i\}_{i=1}^{n}$ will satisfy the complementary slackness condition:

\begin{equation}
\label{cslack}
(Z- p_i \ketbra{\psi_i}{\psi_i}) E_i =0, \, \forall \, 1\leq i \leq n.
\end{equation}

Now summing over $i$ in equation \eqref{cslack} and using the fact that $ \sum_{i=1}^{n} E_i = \mathbb{1}_n$ we get the following.

\begin{equation}
\label{Z}
 Z  = \sum_{i=1}^{n} p_i \ketbra{\psi_i}{\psi_i} E_i = \sum_{i=1}^{n} E_i p_i \ketbra{\psi_i}{\psi_i}.
\end{equation}

From equation \eqref{cslack} we get

\begin{equation}
 % \quad   & E_j ( Z - p_i \ketbra{\psi_i}{\psi_i}) E_i   =   E_j ( Z - p_j \ketbra{\psi_j}{\psi_j}) E_i \notag \\
\label{St}
  E_j \left( p_j \ketbra{\psi_j}{\psi_j} -  p_i \ketbra{\psi_i}{\psi_i} \right) E_i  =0, \; \forall \; 1 \leq i,j \leq n.
\end{equation}

Conditions \eqref{cslack} and \eqref{St} are equivalent to each other. $Z$, given by equation \eqref{Z}, has to satisfy another condition

\begin{equation}
\label{Glb}
 Z \geq p_i \ketbra{\psi_i}{\psi_i}, \;  \forall \; 1  \leq i \leq n.
\end{equation}
 
Thus the \emph{necessary and sufficient} conditions for the $n$-element POVM $\{ E_i \}_{i=1}^{n}$ to maximize $P_s$ are given by conditions \eqref{St} (or \eqref{cslack}) and \eqref{Glb} together. 
\section{Rotationally Invariant Necessary and Sufficient Conditions for MED}
\label{SSGMM}

We wish to obtain the optimal POVM (which is a rank-one projective measurement) for MED of an ensemble $\widetilde{P} = \{p_i,\ketbra{\psi_i}{\psi_i} \}_{i=1}^{n}$, where $\{ \ket{\psi_i} \}_{i=1}^{n}$ is a LI set. Let $ \tket{\psi}{i} \equiv \sqrt{p_i} \ket{\psi_i} , \; \forall \; 1\leq i \leq n$. Since $\{ \tket{\psi}{i} \}_{i=1}^{n}$ is a LI set, corresponding to this set there exists a \emph{unique} set of vectors $\{ \tket{u}{i} \}_{i=1}^{n} \subset \mathcal{H}$ such that\footnote{Given a set of $n$ LI vectors $\left\{ \tket{\psi}{i} \right\}_{i=1}^{n}$ one can obtain the corresponding set of vectors $\{ \tket{u}{i} \}_{i=1}^{n}$ in the following way: fix a basis to work in, arrange $\brat{\psi}{i}$ as rows in a matrix which we call $V$. $V$ is invertible since its rows are LI. The columns of $V^{-1}$ correspond to the vectors $\tket{u}{i}$ in our chosen basis.}:

\begin{equation}
\label{orth}
\tbraket{\psi}{i}{u}{j}= \delta_{ij}, \; \forall \; 1 \leq i,j \leq n.
\end{equation}

Let $G$ denote the gram matrix of $\left\{ \tket{\psi}{i} \right\}_{i=1}^{n}$. The matrix elements of $G$ are hence given by
\begin{equation}
\label{Gram}
G_{ij} = \tbraket{\psi}{i}{\psi}{j}, \; \forall \; 1 \leq i,j \leq n.
\end{equation}

$Tr(G)=1$. Since $\left\{ \tket{\psi}{i} \right\}_{i=1}^{n}$ is a LI set, $G > 0$. The gram matrix corresponding to the set $\{ \tket{u}{i} \}_{i=1}^{n}$ is $G^{-1}$. \emph{Any} orthonormal basis $\{ \ket{v_i} \}_{i=1}^{n}$ of $\mathcal{H}$ can be represented as:

\begin{equation}
\label{v}
\ket{v_i} = \sum_{\substack{j=1}}^{n} \left(G^{\frac{1}{2}} U\right)_{ji} \tket{u}{j},
\end{equation}

where $G^{\frac{1}{2}}$ is the positive square root of $G$ and $U$ is an $n\times n$ unitary matrix. $U$ captures the unitary degree of freedom of the orthonormal basis $\{ \ket{v_i}\}_{i=1}^{n}$. Any such orthonormal basis corresponds to a rank-one projective measurement:

\begin{equation}
\label{btp}
\{ \ket{v_i} \}_{i=1}^{n} \longrightarrow  \{ \ketbra{v_{i}}{v_{i}} \}_{i=1}^{n}.  
\end{equation}

Using this rank-one projective measurment for MED, the average probability of success is given by:

\begin{equation}
\label{PSU}
P_s = \sum_{\substack{i=1}}^{n} | \braket{ \widetilde{ \psi_i}}{v_i}|^{2} = \sum_{\substack{i=1}}^{n} \left| \left(G^\frac{1}{2} U\right)_{ii}\right|^2.
\end{equation}

Let $\{\ketbra{w_i}{w_i} \}_{i=1}^{n}$ be a rank-one projective measurement, which is also a solution for the $n$-POVM $\{ E_i \}_{i=1}^{n}$ in condition \eqref{St}. Here $\braket{w_i}{w_j}=\delta_{ij}$ for $ i,j = 1, 2, \cdots, n$. Let an $n \times n$ unitary matrix $W$ be related to the rank-one projective measurement $\{\ketbra{w_i}{w_i} \}_{i=1}^{n}$ in the following way.

\begin{equation}
\label{w}
\ket{w_i} = \sum_{\substack{j=1}}^{n} \left(G^{\frac{1}{2}} W\right)_{ji} \tket{u}{j}.
\end{equation}

The unitary matrix $W$ is fixed upto right-multiplication with a diagonal unitary matrix, which changes the phases of $\ket{w_i}$. We will soon fix the phases of $\ket{w_i}$ which will ensure that $W$ will be \emph{unique}.

Thus equation \eqref{St} can be rewritten as:

\begin{equation} 
\label{St1}
\bra{w_j} \left( \ketbrat{\psi}{j}{\psi}{j} - \ketbrat{\psi}{i}{\psi}{i} \right) \ket{w_i} = 0, \; \forall \; 1 \leq i,j \leq n.
\end{equation}

Using equation \eqref{w} in equation \eqref{St1}:

\begin{equation}
\label{StG}
 \left(G^{\frac{1}{2}} W\right)_{jj}^{*} \left(G^{\frac{1}{2}} W\right)_{ji} =  \left(G^{\frac{1}{2}} W\right)_{ij}^{*} \left(G^{\frac{1}{2}} W\right)_{ii}, \; \forall \; 1 \leq i,j \leq n.
\end{equation}

The diagonal elements of the matrix $G^\frac{1}{2}W$ can be made non-negative by appropriately fixing the phases of the $\ket{w_i}$ vectors in the following way: right-multiply $W$ with a diagonal unitary $W'$, whose diagonal elements will be phases. From equation \eqref{w} it is seen that right-multiplying $W$ with $W'$ merely changes the phases of the ONB vectors $\ket{w_i}$, and that they will still satisfy equation \eqref{St1}. We can vary the phases in $W'$ so that the diagonals of $G^\frac{1}{2} W W'$ are non-negative. We absorb $W'$ into $W$. After this absorption, the $n \times n$ unitary $W$ which is associated with the rank-one projective measurement $\{ \ketbra{w_i}{w_i} \}_{i=1}^{n}$, is unique. Continuing, we see that equations \eqref{StG} now take the following form.

\begin{equation}
\label{StG1}
 \left(G^{\frac{1}{2}} W\right)_{jj} \left(G^{\frac{1}{2}} W\right)_{ji} = 
 \left(G^{\frac{1}{2}} W\right)_{ii} \left(G^{\frac{1}{2}} W\right)_{ij}^{*}, \; \forall \; 1 \leq i,j \leq n.
\end{equation}

Let $D= Diag(d_{11},d_{22}, \cdots, d_{nn})$ be the real diagonal matrix of $G^\frac{1}{2} W $, i.e., 
 
\begin{equation}
\label{D}
d_{ii} =  \left(G^{\frac{1}{2}} W \right)_{ii}, \; \forall \; 1 \leq i \leq n.
\end{equation}

From equation \eqref{StG1} and the fact that the diagonals of $G^\frac{1}{2}W$ are all real, we infer that the matrix $DG^{\frac{1}{2}}W$ is hermitian.

\begin{equation}
\label{DGWhermit}
D G^{\frac{1}{2}} W - W^\dag G^{\frac{1}{2}} D = 0. 
\end{equation}

Left multiplying the LHS and RHS by $D G^{\frac{1}{2}} W$ gives

\begin{align}
\label{EOY} %equation of the year :D
& \left(D G^{\frac{1}{2}} W \right)^2 - D G D = 0 \notag \\
\Longrightarrow \; \; & X^2 - DGD  = 0,
\end{align}

where $X \equiv DG^\frac{1}{2}W$, $X^\dag = X$ and (note that) $D^2$ is the diagonal of $X$. 

In the MED problem, we are given the gram matrix $G$ of the ensemble $\widetilde{P}$. To solve condition \eqref{St} for MED of $\widetilde{P}$ we need to solve for $X$ in equation \eqref{EOY}. Knowing $X$ tells us what $G^\frac{1}{2}W$ is, which can be used in equation \eqref{w} to obtain $\{\ketbra{w_i}{w_i}\}_{i=1}^{n}$. Equation \eqref{EOY} came from assuming that $\{ \ketbra{w_i}{w_i} \}_{i=1}^{n}$ represented some $n$-POVM which satisfied condition \eqref{St}. For $\{ \ketbra{w_i}{w_i} \}_{i=1}^{n}$ to be the optimal POVM it needs to satisfy condition \eqref{Glb} too; this will impose another condition on the solution for $X$ in equation \eqref{EOY}. 

\begin{theorem}
\label{thmret}
Let $\mathcal{X}$ be a solution for $X$ in equation \eqref{EOY}. Then $\mathcal{X}$ corresponds to the optimal POVM for MED of $\widetilde{P}$ if it is positive definite. Also, $\mathscr{R}_\mathcal{G}(G) = \dfrac{1}{Tr(D^2G)} DGD$, where $D$ is the square root of the diagonal of $\mathcal{X}$.
\end{theorem}

\begin{proof}
We relate $d_{ii}$, defined in equation \eqref{D}, to the probability $q_i$ mentioned in equation \eqref{piqi}. In section \eqref{intro} it was mentioned that the optimal POVM for MED of $\widetilde{P}$ is the PGM of an ensemble $\mathscr{R}\left(\widetilde{P}\right)=\widetilde{Q} = \{q_i, \ketbra{\psi_i}{\psi_i} \}_{i=1}^{n}$ (see definition \eqref{mathRdef}). This means that \cite{Bela}

\begin{equation}
\label{pretty}
\ket{w_i} = \left( \sum_{j=1}^{n} \ketbrat{\psi'}{j}{\psi'}{j} \right)^{-1/2} \tket{\psi'}{i}, \; \forall \; 1 \leq i \leq n,
\end{equation}

where $\tket{\psi'}{i}\equiv \sqrt{q_i}\ket{\psi_i}$ and $\left( \sum_{j=1}^{n} \ketbrat{\psi'}{j}{\psi'}{j} \right)^{-1/2}  >0$. Define $\tket{u'}{i}$ to be such that  $\tbraket{\psi'}{i}{u'}{j}= \delta_{ij}, \; \forall \; 1 \leq i,j \leq n$. $G_q$ is the gram matrix corresponding to the ensemble $\widetilde{Q}$. It can be verified that $G_q^{-1}$ is the gram matrix of the vectors $\{ \tket{u'}{i} \}_{i=1}^{n}$, i.e., $\left(G_q^{-1}\right)_{ij} = \tbraket{u'}{i}{u'}{j}, \; \forall \; 1 \leq i,j \leq n$. This implies that 

\begin{equation}
\label{rho'}
\left( \sum_{j=1}^{n} \ketbrat{\psi'}{j}{\psi'}{j} \right)^{-1/2} = \left( \sum_{j=1}^{n} \ketbrat{u'}{j}{u'}{j} \right)^{1/2} = \sum_{i,j=1}^{n} \left( G_q^\frac{1}{2} \right)_{ij} \ketbrat{u'}{i}{u'}{j}.
\end{equation}

Note that since the LHS in equation \eqref{rho'} is positive definite, the RHS in equation \eqref{rho'} should also be positive definite and that can only be true if $G_q^\frac{1}{2} > 0$. One can verify the above equation by squaring on both sides \footnote{That $\tket{\psi'}{i}$ and $\tket{u'}{j}$ are related by $\tbraket{\psi'}{i}{u'}{j}= \delta_{ij}$ implies that $\sum_{j=1}^{n} \ketbrat{u'}{j}{\psi'}{j} = \mathbb{1}_n$. This can be seen from the fact that if $\ket{\eta} = \sum_{j=1}^{n} \alpha_j \tket{u'}{j}$ is any vector in $\mathcal{H}$, then $\left(\sum_{j=1}^{n} \ketbrat{u'}{j}{\psi'}{j}\right) \ket{\eta} = \ket{\eta}$. That $\sum_{j=1}^{n} \ketbrat{u'}{j}{\psi'}{j} = \mathbb{1}_n$ is true implies that $\left(\sum_{j=1}^{n} \ketbrat{u'}{j}{u'}{j}\right)\left(\sum_{k=1}^{n} \ketbrat{\psi'}{k}{\psi'}{k}\right) = \mathbb{1}_n$. Hence  $\sum_{j=1}^{n} \ketbrat{u'}{j}{u'}{j}$ is the inverse of $\sum_{k=1}^{n} \ketbrat{\psi'}{k}{\psi'}{k}$. Also, since $G_q$ is the gram matrix of $\left\{ \tket{\psi'}{j} 
\right\}_{j=1}^n$ and $G_q^{-1}$ is the gram matrix of $\left\{ \tket{u'}{j} \right\}_{j=1}^{n}$ we get that $\left(  \sum_{i,j=1}^{n} \left( G_q^\frac{1}{2} \right)_{ij} \ketbrat{u'}{i}{u'}{j} \right)^2 =  \sum_{j=1}^{n} \ketbrat{u'}{j}{u'}{j}$.}. Substituting the above in equation \eqref{pretty} gives 

\begin{eqnarray}
\label{wu'1}
& \ket{w_i} & = \sum_{j=1}^{n} \left(G_q^\frac{1}{2} \right)_{ji} \tket{u'}{j} \notag \\
 & ~ & = \sum_{j=1}^{n}  \frac{\sqrt{p_j}}{\sqrt{q_j}} \left(  G_q^\frac{1}{2} \right)_{ji}  \tket{u}{j}, \; \forall \; 1 \leq i \leq n,
\end{eqnarray}

where $ \tket{u'}{i} = \frac{\sqrt{p_i}}{\sqrt{q_i}} \tket{u}{i}$, $\forall \; 1 \leq i \leq n$ (since $ \tket{\psi'}{i} = \frac{\sqrt{q_i}}{\sqrt{p_i}} \tket{\psi}{i} $). Since $\{ \tket{u}{i} \}_{i=1}^{n}$ is a basis for $\mathcal{H}$, on comparing equations \eqref{wu'1} and \eqref{w} we obtain

\begin{subequations}
\begin{equation}
\label{atd}
 \left( G^{\frac{1}{2}} W \right)_{ji}  =  \dfrac{\sqrt{p_j}}{\sqrt{q_j}} \left(  G_q^\frac{1}{2} \right)_{ji}, \; \forall \; 1 \leq i,j, \leq n, 
\end{equation}
\begin{equation}
 \label{aj}
\Longrightarrow  \; \; d_{jj}  =  \dfrac{\sqrt{p_j}}{\sqrt{q_j}} \left(  G_q^\frac{1}{2} \right)_{jj}, \; \forall \; 1 \leq j \leq n.
\end{equation}
\end{subequations}

Using equation \eqref{piqi} we get that 

\begin{equation}
\label{aipiqi}
d_{jj}\dfrac{\sqrt{p_j}}{\sqrt{q_j}} = \frac{p_j}{q_j} \left(  G_q^\frac{1}{2} \right)_{jj} = C, \; \forall \; 1 \leq j \leq n,
\end{equation}

where $C$ is the positive constant that appears in equation \eqref{piqi}. Since $d_{jj} \frac{\sqrt{p_j}}{\sqrt{q_j}} = C, \; \forall \;  1 \leq j \leq n$, using equations \eqref{atd} and \eqref{aipiqi}  we get that $$(\mathcal{X})_{ji} = \left(DG^\frac{1}{2}W\right)_{ji} = d_{jj}\left(G^\frac{1}{2}W \right)_{ji}= C \times \left(G_q^\frac{1}{2} \right)_{ji},  \; \forall \; 1 \leq i,j, \leq n,$$ that is, $\mathcal{X}$ is equal to the product of a positive constant $C$ and $G_q^\frac{1}{2}$, which implies that  $\mathcal{X} > 0$. Also from equation \eqref{EOY} it follows that $DGD =C^2 \times G_q$, i.e., the gram matrix of $\mathscr{R}\left(\widetilde{P}\right)=\widetilde{Q}$ is given by $\frac{DGD}{Tr(D^2G)}$, i.e., 
\begin{equation}
\label{GXR}
\mathscr{R}_\mathcal{G}(G) = \frac{\mathcal{X}^2}{Tr\left(\mathcal{X}^2\right)} =  \frac{DGD}{Tr(D^2G)}.
\end{equation}
\end{proof}

The converse of theorem \ref{thmret} is proved in the following.

\begin{theorem}
\label{thm2}
If $\mathcal{X}$ is a solution for $X$ in equation \eqref{EOY} and $\mathcal{X}$ is positive definite, then $\mathcal{X}$ corresponds to the optimal POVM for MED of the ensemble $\widetilde{P}$. Also, $\mathcal{X}$ is unique, i.e., there is no other $\mathcal{X}'$ such that it is a solution for $X$ in equation \eqref{EOY} and $\mathcal{X}'>0$.
\end{theorem}

\begin{proof}
Let $\mathcal{X}$ be a solution for $X$ in equation \eqref{EOY} and let $\mathcal{X}$ be positive definite. Equating $D^{-1} \mathcal{X}$ with $G^\frac{1}{2}W$ (see below equation \eqref{EOY}) and employing it in equation \eqref{w}, we obtain $\{ \ketbra{w_i}{w_i} \}_{i=1}^{n}$ to be the rank-one projective measurement corresponding to solution $\mathcal{X}$. We want to prove that $\{ \ketbra{w_i}{w_i} \}_{i=1}^{n}$ is the optimal POVM. For this purpose we borrow a result from Mochon's paper. Equation (33) in Mochon's paper \cite{Carlos} tells us that the positive operator $Z$, defined in equation \eqref{Z}, is a scalar times the positive square root of the density matrix of the ensemble $\mathscr{R}\left(\widetilde{P}\right)$, i.e.,
\begin{equation}
\label{ZQ}
Z = C \left( \sum_{i=1}^{n} q_i \ketbra{\psi_i}{\psi_i} \right)^\frac{1}{2}.
\end{equation}
We will compute $\sum_{i=1}^{n} p_i |w_i\rangle\langle w_i | \psi_i \rangle\langle \psi_i | $ and show that it is equal to $C ( \sum_{i=1}^{n} q_i \ketbra{\psi_i}{\psi_i} )^\frac{1}{2}$, thereby proving that $\sum_{i=1}^{n} p_i |w_i\rangle\langle w_i | \psi_i \rangle\langle \psi_i | $ is equal to $Z$. This will then imply that $\{ \ketbra{w_i}{w_i} \}_{i=1}^{n}$ is the optimal POVM.

Since $\ket{w_i} = \sum_{k=1}^{n} (D^{-1}\mathcal{X})_{ki} \tket{u}{k}$ and $\tket{u}{k} = \sum_{j=1}^{n} (G^{-1})_{jk} \tket{\psi}{j}$, using equation \eqref{EOY} it's easily verified that $\ket{w_i} = \sum_{j=1}^{n} \left( D\mathcal{X}^{-1} \right)_{ji} \tket{\psi}{j}$. Then

\begin{equation}
\label{Zw}
\sum_{i=1}^{n} |w_i\rangle\langle w_i | \widetilde{\psi}_i \rangle\langle \widetilde{\psi}_i |  =   \sum_{j,i =1}^{n} \left( D \mathcal{X}^{-1} D \right)_{ji} \ketbrat{\psi}{j}{\psi}{i} > 0.
\end{equation}

Squaring the RHS in equation \eqref{Zw} and employing equation \eqref{EOY} we get that 

\begin{equation}
\label{Z2}
\left(\sum_{i,j =1}^{n} \left( D \mathcal{X}^{-1} D \right)_{ij} \ketbrat{\psi}{i}{\psi}{j}\right)^2 = \sum_{i,j =1}^{n} d_{ii}^2 \ketbrat{\psi}{i}{\psi}{i}.
\end{equation}

Consider the probability $ k_i \equiv \dfrac{d_{ii}^2p_i}{\sum_{j=1}^{n}d_{jj}^2p_j}, \; \forall \; 1 \leq i \leq n$. Thus $\sum_{i=1}^n k_i \ketbra{\psi_i}{\psi_i}$ is the average state of the ensemble $\widetilde{K} = \{ k_i, \ketbra{\psi_i}{\psi_i} \}_{i=1}^{n}$. The matrix elements of the gram matrix, $G_k$ of $\widetilde{K}$ are then given by $$(G_k)_{ij} = \sqrt{k_i k_j} \braket{\psi_i}{\psi_j} = \dfrac{1}{\sum_{l=1}^{n}d_{ll}^2p_l}d_{ii} \tbraket{\psi}{i}{\psi}{j} d_{jj}.$$ This tells us that $G_k=\frac{1}{Tr(D^2G)}DGD$; using equation \eqref{EOY} we get that the positive square root of $G_k$ is $G_k^\frac{1}{2}=\frac{1}{\sqrt{Tr(D^2G)}}\mathcal{X}$ and, hence, $d_{ii}^2 =\mathcal{X}_{ii}= \sqrt{Tr(D^2G)}\left(G_k^\frac{1}{2}\right)_{ii}$ (see equation \eqref{D}). Thus $k_i$ and $p_i$ are related by the equation 

\begin{equation}
\label{aiki}
p_i = C' \dfrac{k_i}{ \left( G_k^\frac{1}{2} \right)_{ii}}, \; \forall \; 1 \leq i \leq n,
\end{equation}

where $C'$ is the normalization constant given by $C'= \sqrt{Tr(D^2G)}$. We see that $p_i$ are related to the $k_i$ in the exact same way that $p_i$ are related to $q_i$ from equation \eqref{piqi}. Below definition \eqref{mathRR}, it was mentioned that if $\widetilde{P}$ and $\widetilde{K}$ are two ensembles with the same states with apriori probabilities $p_i$ and $k_i$, which are related by equation \eqref{aiki}, we get that $\mathscr{R}^{-1}\left(\widetilde{K}\right) = \widetilde{P}$. Since $\mathscr{R}^{-1}$ is a bijection, this implies that $\widetilde{K} = \mathscr{R}\left(\widetilde{P}\right)= \widetilde{Q}$ and $k_i = q_i, \; \forall \; 1 \leq i \leq n$, where $q_i$ is the apriori probability of states in $\widetilde{Q}$ as given in equation \eqref{piqi}. This also implies that $C' = C$. 

From equation \eqref{ZQ} we get that the RHS of equation \eqref{Z2} equates to 
$$   \left( \sum_{i=1}^{n} |w_i\rangle\langle w_i | \widetilde{\psi}_i \rangle\langle \widetilde{\psi}_i |\right)^2 =  \sum_{i=1}^n d_{ii}^2 \ketbrat{\psi}{i}{\psi}{i} = C^2\left( \sum_{i=1}^{n} q_i\ketbra{\psi_i}{\psi_i} \right) = Z^2.$$
Then the fact that $\sum_{i=1}^{n} |w_i\rangle\langle w_i | \widetilde{\psi}_i \rangle\langle \widetilde{\psi}_i | $ is positive definite tells us that
\begin{equation}
\label{Zw1}
\sum_{i=1}^{n} |w_i\rangle\langle w_i | \widetilde{\psi}_i \rangle\langle \widetilde{\psi}_i | = C \left( \sum_{i=1}^{n} q_i \ketbra{\psi_i}{\psi_i} \right)^\frac{1}{2} = Z .
\end{equation} 

Note that the ONB $\{ \ket{w_i}\}_{i=1}^{n}$ was constructed from $\mathcal{X}$, which solves for $X$ in equation \eqref{EOY} and which is positive definite. That $\sum_{i=1}^{n} |w_i\rangle\langle w_i | \widetilde{\psi}_i \rangle\langle \widetilde{\psi}_i |$ is equal to $C \left( \sum_{i=1}^{n} q_i \ketbra{\psi_i}{\psi_i} \right)^\frac{1}{2} $, which \emph{we already know} is equal to $Z$ \cite{Carlos}, implies that $\{ \ketbra{w_i}{w_i} \}_{i=1}^{n}$ is the optimal POVM.  

Thus $\{ \ketbra{w_i}{w_i} \}_{i=1}^{n}$ is the optimal POVM. 

Since $\{ \ketbra{w_i}{w_i} \}_{i=1}^{n}$ is the \emph{unique} optimal POVM for MED of $\widetilde{P}$ so is the $n$ tuple $(d_{11},d_{22},\cdots,d_{nn})$ unique to the MED of $\widetilde{P}$, \footnote{Note that $d_{ii} = \braket{w_i}{\widetilde{\psi}_i}$, thus if $\{\ketbra{w_i}{w_i}\}_{i=1}^{n}$ is unique, so must the $n$-tuple $(d_{11},d_{22},\cdots,d_{nn})$.}. This implies that $D=Diag(d_{11},d_{22},\cdots,d_{nn})$ is unique, which implies that $DGD$ is unique and since the positive square root of $DGD$ is also unique, that tells us that $\mathcal{X}$ is unique too. \end{proof}

Theorem \eqref{thm2} tells us that for MED of any $\widetilde{P} \in \mathcal{E}$, $\mathcal{X}$ is unique. But note that if $\ket{\psi_i}$ underwent a rotation by unitary $U$ then it can be inferred from equation \eqref{EOY} that the solution for $\mathcal{X}$ won't change since $G$ doesn't change. This implies that $\mathcal{X}$ is a function of $G$ in $\mathcal{G}$.

Let the matrix elements of $\mathcal{X}$ be given by the following equation

\begin{equation}
\label{mathcalW}
\mathcal{X} =  \begin{pmatrix}
                      {d_{11}}^2 & d_{12} + i d_{21} & d_{13} + i d_{31} & \cdots &  d_{1n}  + i d_{n1}  \\
  d_{12}  - i d_{21} & {d_{22}}^{2} & d_{23}  + i d_{32}  & \cdots &  d_{2n}  + i d_{n2}\\
d_{13} - i d_{31}  & d_{23}  - i d_{32}   & {d_{33}}^{2} & \cdots &  d_{3n}  + i d_{n3}  \\
\vdots & \vdots & \vdots & \ddots & \vdots \\
d_{1n}  - i d_{n1}  & d_{2n}  - id_{n2}   & d_{3n} - i d_{n3}  & \cdots &  {d_{nn}}^2 
                     \end{pmatrix},
\end{equation}

where $d_{kl}$ are the real and imaginary parts of the matrix elements of $\mathcal{X}$. Since $\mathcal{X}$ is a function on $\mathcal{G}$, $d_{kl}$ are also functions on $\mathcal{G}$. 

\begin{mydef}
\label{defQ}
Let $\mathcal{Q}$ denote the set of all positive definite $n \times n$ matrices. 
\end{mydef}
Thus $\mathcal{G} \subset \mathcal{Q}$. Using $\mathcal{G}$ and $\mathcal{Q}$ we formalize $\mathcal{X}$ as a function on $\mathcal{G}$. 

\begin{mydef}
\label{funmathD} 
\label{alphai} 
\label{funW} 
\label{defrhokl}
\begin{subequations}
$\mathcal{X}: \mathcal{G} \longrightarrow \mathcal{Q}$ is such that $\mathcal{X}(G)$ solves equation \eqref{EOY}
\begin{equation}
\label{calEOY}
\left(\mathcal{X}(G)\right)^2 - D(G) \, G  \,  D(G) = 0.
\end{equation}
Let's denote $d_{kl}: \mathcal{G} \longrightarrow \mathbb{R}$ to be the real and imaginary parts of matrix elements of $\mathcal{X}(G)$, i.e., 
\begin{align}
\label{rhodef}
 & ~~~~~~~~~~~~~~~~~~~~~~~~ d_{kl}\left(G\right) \equiv Re \left( \left(  \mathcal{X}\left(G\right) \right)_{kl} \right), \; \forall \; 1 \leq k < l \leq n, \\
 & ~~~~~~~~~~~~~~~~~~~~~~~~ d_{ii}\left(G\right) \equiv  \sqrt{\left(  \mathcal{X}\left(G\right) \right)_{ii}}, \; \forall \; 1 \leq i \leq n, \\
 & ~~~~~~~~~~~~~~~~~~~~~~~~ d_{kl}\left(G\right) \equiv Im \left( \left(  \mathcal{X}\left(G\right) \right)_{lk} \right), \; \forall \; 1 \leq l < k \leq n,
\end{align} and $D(G) \equiv Diag(d_{11}(G),d_{22}(G), \cdots, d_{nn}(G))$. 

\end{subequations}
\end{mydef}

Note that if one knows the real $n$-tuple $(d_{11}(G),d_{22}(G), \cdots, d_{nn}(G))$, then using equation \eqref{EOY} one can compute $\mathcal{X}(G)$. Thus we have reformulated the MED problem for linearly independent pure states in a rotationally invariant way: \bigskip

\textbf{Rotationally Invariant Necessary and Sufficient Conditions:} \emph{ Let $G$ be the gram matrix corresponding to an $n$-LIP: $\{ p_i, \ketbra{\psi_i}{\psi_i} \}_{i=1}^{n}$. To solve MED for this $n$-LIP, one needs to find real and positive $n$-tuple }$(d_{11}(G),d_{22}(G), \cdots, d_{nn}(G))$\emph{ such that, when arranged in the diagonal matrix }$ D(G) = Diag(d_{11}(G),d_{22}(G), \cdots, d_{nn}(G))$, \emph{the diagonal of the positive square root of }$ D(G) G D(G)$\emph{is} $\left( D(G)\right)^2$.

\section{\label{empIFT} Solution for the MED of LI Pure State Ensembles}

$\mathcal{X}$ is a function on $\mathcal{G}$ such that $\mathcal{X}(G)$ is a solution for $X$ in equation \eqref{EOY}, and is positive definite. We need to compute $\mathcal{X}(G)$ for a given $G \in \mathcal{G}$. We employ the Implicit Function Theorem (IFT) for this. 

\subsection{Functions and Variables for IFT}
\label{defFV}

In this subsection, we will introduce the functions and variables which are part of the IFT technique.

Let the unknown  hermitian matrix $X$ in equation \eqref{EOY} be represented by
\begin{equation}
\label{DGWlooks1}
X = \begin{pmatrix}
     x_{11}^2 & x_{12} + i x_{21} & x_{13} + i x_{31} & \cdots &  x_{1n}  + i x_{n1}  \\
  x_{12}  - i x_{21}  & x_{22}^2 & x_{23}  + i x_{32}  & \cdots &  x_{2n}  + i x_{n2}  \\
x_{13} - i x_{31}  & x_{23}  - i x_{32}   & x_{33}^2 & \cdots &  x_{3n}  + i x_{n3}  \\
\vdots & \vdots & \vdots & \ddots & \vdots \\
x_{1n}  - i x_{n1}  & x_{2n}  - ix_{n2}   & x_{3n} - i x_{n3}  & \cdots &  x_{nn}^2
                     \end{pmatrix},
\end{equation}

where $x_{kl} \in \mathbb{R}, \; \forall \; 1 \leq k, l \leq n $. 

Define $F$ on $\mathcal{G} \times \mathcal{H}_n$, where $\mathcal{H}_n$ is the real vector space of all $n \times n$ hermitian matrices.

\begin{mydef}
\begin{equation}
\label{defF}
F(G,X) \equiv X^2 - D(X) G D(X),
\end{equation}
where $X$ is of the form given in equation \eqref{DGWlooks1} and $D(X) \equiv Diag \left(x_{11},x_{22},\cdots,x_{nn}\right)$.
\end{mydef}

We define the matrix elements of $F$ as functions of $G$ and $x_{ij},$ $\forall\; 1 \le i,j \le n$.

\begin{equation}
\label{f}
F = \begin{pmatrix}
        f_{11} & f_{12} + i f_{21} & f_{13} + i f_{31} & \cdots &  f_{1n}  + i f_{n1}  \\
  f_{12}  - i f_{21}  & f_{22} & f_{23}  + i f_{32}  & \cdots &  f_{2n}  + i f_{n2}  \\
f_{13} - i f_{31}  & f_{23}  - i f_{32}   & f_{33} & \cdots &  f_{3n}  + i f_{n3}  \\
\vdots & \vdots & \vdots & \ddots & \vdots \\
f_{1n}  - i f_{n1}  & f_{2n}  - if_{n2}   & f_{3n} - i f_{n3}  & \cdots &  f_{nn}
                     \end{pmatrix},
\end{equation}

where, for $i<j$, $f_{ij}$ and $f_{ji}$ represent the real and imaginary parts of $F_{ij}$ respectively, and $f_{ii}$ represents the diagonal element $F_{ii}$, and for $j < i$, $f_{ji}$ and $-f_{ij}$ represent the real and imaginary parts of $F_{ij}$ respectively. Then using definition \eqref{defF} and equation \eqref{DGWlooks1} we get for $ i< j$

\begin{subequations}
\label{feq}
\begin{align}
\label{fij}
f_{ij}\left(G, \vv{x}\right) = \; & \sum_{k=1}^{i-1} \left(x_{ki}x_{kj} + x_{ik}y_{jk} \right) + \sum_{k=i+1}^{j-1} \left(x_{ik}x_{kj} - x_{ki}x_{jk} \right) \notag  +   \sum_{k=j+1}^{n} \left(x_{ik}x_{jk} + x_{ki}x_{kj} \right) \\ + & x_{ij}\left(x_{ii}^2 + x_{jj}^2\right) - x_{ii}x_{jj} Re\left(G_{ij}\right),
\end{align}
\begin{align}
\label{fji}
f_{ji}\left(G, \vv{x}\right) = \; & \sum_{k=1}^{i-1} \left(x_{ki}x_{jk} - x_{ik}x_{kj} \right) + \sum_{k=i+1}^{j-1} \left(x_{ik}x_{jk} + x_{ki}x_{kj} \right) \notag  +   \sum_{k=j+1}^{n} \left(-x_{ik}x_{kj} + x_{ki}x_{jk} \right) \\ + & x_{ji}\left(x_{ii}^2 + x_{jj}^2\right) - x_{ii}x_{jj} Im\left(G_{ij}\right),
\end{align}
and for the diagonal elements we get

\begin{align}
\label{fii}
f_{ii}\left(G, \vv{x}\right) = \; & \sum_{k=1}^{i-1} \left(x_{ki}^2 + x_{ik}^2 \right) + \sum_{k=i+1}^{n} \left(x_{ik}^2 + x_{ki}^2 \right) + 2 x_{ii}^4 - x_{ii}^2G_{ii},
\end{align}
where $\vv{x} \equiv (x_{11},x_{12},\cdots, x_{nn})$ (i.e., $\vv{x}$ is the real $n^2$-tuple of the $x_{ij}$-variables). 
\end{subequations}

Finally, we define the Jacobian of the functions $f_{ij}$ with respect to the variables $x_{ij}$; this Jacobian matrix has the matrix elements

\begin{equation}
\label{Jacobian}
\left(J \left( G, \vv{x} \right)\right)_{ij,kl} \equiv \dfrac{\partial{f_{ij}   \left(G,\vv{x} \right) }}{\partial{x_{kl}}}, \; \forall \; 1 \le i,j,k,l \le n.
\end{equation}

Note that since the $f_{ij}$ functions and the $x_{ij}$ variables are both $n^2$ in number, this Jacobian matrix is an $n^2 \times n^2$ square matrix.

\subsection{Implementing IFT}

Let $G_0 \in \mathcal{G}$ be a gram matrix whose MED for which we know the solution, that is, we know the values of $x_{ij} = d_{ij}(G_0)$, $\forall$ $1 \le i,j \le n$ (see definition \eqref{defrhokl}). Substituting $x_{ij} = d_{ij}(G_0)$, $\forall$ $1 \le i,j \le n$ in equation \eqref{DGWlooks1} gives us that $X = \mathcal{X}(G_0)$ (see equations \eqref{mathcalW} and definition \eqref{funmathD}), and substituting $X = \mathcal{X}(G_0)$ into equation \eqref{defF} gives (see theorem \ref{thm2}), 
\begin{enumerate}                                                                                                                                                                                                                                                                                                               
\item[(i.)] $ F \Big(G_0, \; X = \mathcal{X}(G_0) \; \Big) = 0$. This equation tells us that $X = \mathcal{X}(G_0)$ is a solution for $X$ in equation \eqref{EOY} when $G = G_0$.
\item[(ii.)] $X = \mathcal{X}(G_0) > 0$.                                                                                                                                                                                                                                                                                                                 \end{enumerate}
 
IFT, which is a well known result in functional analysis \cite{imp}, then tells us the following.

\paragraph{Implicit Function Theorem:} Consider the following inequality:
\begin{equation}\label{Detjac} Det \left( J(G_0,\vv{d}(G_0)) \right) \neq 0, \; \text{where} \; \vv{d}(G_0) = \left(d_{11}(G_0),d_{12}(G_0),\cdots,d_{1n}(G_0),d_{21}(G_0),\cdots,d_{nn}(G_0)\right).\end{equation} 
If the inequality \eqref{Detjac} is true, then IFT tells us that there exists an open neighbourhood $I_{G_0}$ in $\mathcal{G}$ containing $G_0$, such that for each $i,j$, where $1 \le i,j \le n$, there exists an open interval $I_{ij}$ in $\mathbb{R}$ containing the real number $d_{ij}(G_0)$, such that one can define the function $\phi_{ij}: I_{G_0} \longrightarrow I_{ij}$, such that
\begin{enumerate}
 \item $\phi_{ij}$'s are continuously differentiable in $I_{G_0}$,
 \item $\phi_{ij}(G_0) = d_{ij}(G_0)$, $\forall \; 1 \le i,j \le n$, and
 \item the following equation holds true for $\; \forall \; 1 \le i,j \le n$ and $\forall \; G \in I_{G_0}$:
 \begin{equation}
\label{impff1}
f_{ij} \left( G, \vv{\phi}(G) \right) = 0, \text{ where } \vv{\phi}(G) = ( \phi_{11}(G),\phi_{12}(G),\cdots, \phi_{nn}(G)).
\end{equation}
\end{enumerate}

Thus to use the IFT for our purpose we need to prove the following. 

\begin{theorem}
\label{Jacthm}
$Det \left( J\left( G,\vv{d}(G)\right)\right) \neq 0, \; \forall \; G \in \mathcal{G}$.
\end{theorem}

\begin{proof} This proof is divided into two parts:
\begin{enumerate}
 \item[(a.)] To show that $J(G,\vv{d}(G))$ is a linear transformation on the real space $\mathcal{H}_n$ of $n \times n$ complex hermitian matrices: \medbreak 
 Proof of (a.): \begin{subequations}First note that $J(G,\vv{d}(G))$ is the Jacobian of the function $F$ with respect to the variable $X$ (equation \eqref{defF}). 
 
 Let $x_{ij}$ be assigned the value $d_{ij}(G)$ for all $1 \le i, j \le n$. Now let $x_{ij}=d_{ij}(G) \longrightarrow x_{ij}=d_{ij}(G) + \epsilon \delta x_{ij}$ be an arbitrary perturbation, where $\epsilon$ is an infinitesimal positive real number and $\delta x_{ij}$ are real, $\forall$ $1 \le i,j \le n$. As a result of this perturbation we have the following 
 \begin{enumerate}
 \item[(i)] $\left(x_{ii}(G)\right)^2 = $  $\left(d_{ii}(G)\right)^2 \longrightarrow \left(d_{ii}(G)\right)^2 + 2 \epsilon d_{ii}(G) \delta x_{ii} + \mathcal{O}(\epsilon^2)$, and 
 \item[(ii)] $X = \mathcal{X}(G) \longrightarrow X = \mathcal{X}(G) + \epsilon \delta X + \mathcal{O}(\epsilon^2)$ where 
 \begin{equation}
\label{DeltaX}
\delta X = \begin{pmatrix}
        2 d_{11}(G)\delta x_{11} & \delta x_{12} + i \delta x_{21} & \delta x_{13} + i \delta x_{31} & \cdots &  \delta x_{1n}  + i \delta x_{n1}  \\
  \delta x_{12}  - i \delta x_{21}  & 2 d_{22}(G) \delta x_{22} & \delta x_{23}  + i \delta x_{32}  & \cdots &  \delta x_{2n}  + i \delta x_{n2}  \\
\delta x_{13} - i \delta x_{31}  & \delta x_{23}  - i \delta x_{32}   & 2 d_{33}(G) \delta x_{33} & \cdots &  \delta x_{3n}  + i \delta x_{n3}  \\
\vdots & \vdots & \vdots & \ddots & \vdots \\
\delta x_{1n}  - i \delta x_{n1}  & \delta x_{2n}  - i\delta x_{n2}   & \delta x_{3n} - i \delta x_{n3}  & \cdots &  2 d_{nn}(G) \delta x_{nn}                     \end{pmatrix}.
\end{equation}

For the sake of brevity, for the rest of this proof, we will denote $D(G)$ by $D$, $\mathcal{X}(G)$ by $\mathcal{X}$, and $J(G,\vv{d}(G))$ by $J_G$. Define:

\begin{equation}
\label{Ddelta}
D_\delta \equiv Diag(\delta x_{11},\delta x_{22},\cdots,\delta x_{nn}).
\end{equation}
\end{enumerate}

Thus we get the following. 
\begin{align}
\label{Fpert}
&      F\left(G, \mathcal{X} + \epsilon \delta X\right) \notag \\
= \; & F\left(G,\mathcal{X}\right) + \epsilon \left( \delta X \mathcal{X} + \mathcal{X} \delta X -  D_{\delta } G  D - D G D_{\delta } \right) + \mathcal{O}(\epsilon^2)  \notag \\
= \; & \epsilon \left( \delta X \mathcal{X} - D_{\delta }D^{-1}\mathcal{X}^2 + \mathcal{X} \delta X -  \mathcal{X}^2 D^{-1} D_{\delta } \right) + \mathcal{O}(\epsilon^2) \notag \\ 
= \; & \epsilon J_G \left(\delta X \right)+ \mathcal{O}(\epsilon^2),
\end{align}

where equation \eqref{EOY} was employed in the second step above, and

\begin{equation}
\begin{split}
J_G\left( \delta X \right) & = \; \delta X \mathcal{X} - D_{\delta }D^{-1}\mathcal{X}^2 + \mathcal{X} \delta X -  \mathcal{X}^2 D^{-1} D_{\delta }\\                              
\label{J}
& = \;  \left(\delta X \mathcal{X} - D_{\delta }D^{-1}\mathcal{X}^2 \right) + \left(\delta X \mathcal{X} - D_{\delta }D^{-1}\mathcal{X}^2 \right)^\dag.
\end{split}
\end{equation}

Thus it is seen that $J_G$ is a linear transformation on the space of $n \times n$ complex hermitian matrices $\mathcal{H}_n$. 
\end{subequations}
 
 \item[(b.)] If the action of $J(G,\vv{d}(G))$ on some $n \times n$ complex hermitian matrix $\delta X$ is $0$ then $\delta X$ itself must be $0$.\medbreak
 Proof of (b.): From equation \eqref{J} it is clear that $J_G\left( \delta X \right) =0$ if and only if $\delta X \mathcal{X} - D_{\delta }D^{-1}\mathcal{X}^2$ is anti-hermitian.
\begin{subequations}
Let's assume that $\delta X \mathcal{X} - D_{\delta }D^{-1}\mathcal{X}^2$ is anti-hermitian. That is, 
\begin{align}
\label{ij1}
& \delta X \mathcal{X} - D_{\delta }D^{-1} \mathcal{X}^2 =  - \mathcal{X} \delta X  +   \mathcal{X}^2 D^{-1} D_{\delta } \notag \\
\Longrightarrow \; & \mathcal{X}^{-1} \delta X  - \mathcal{X}^{-1}  D_{\delta }D^{-1} \mathcal{X} =  -  \delta X \mathcal{X}^{-1}  + \mathcal{X} D^{-1} D_{\delta } \mathcal{X}^{-1}.
\end{align}

Let $\mathcal{X} = \sum_{i=1}^{n} g_i \ketbra{g_i}{g_i}$ be the spectral decomposition of $\mathcal{X}$. Then the $ij$-th matrix element of the matrix in equation \eqref{ij1}, in the $\{ \ket{g_i}\}_{i=1}^{n}$ basis, is given by
\begin{align*}
& \dfrac{1}{g_i} \pk{g_i}{\delta X}{g_j} - \dfrac{g_j}{g_i} \pk{g_i}{D_{\delta } D^{-1}}{g_j} \notag \\ = \;  - & \dfrac{1}{g_j} \pk{g_i}{\delta X}{g_j} + \dfrac{g_i}{g_j} \pk{g_i}{D_{\delta } D^{-1}}{g_j} 
\end{align*}
\begin{equation}
\label{ij2}
\Longrightarrow \pk{g_i}{\delta X}{g_j} = \left( \dfrac{g_i^2 + g_j^2}{g_i + g_j} \right)\pk{g_i}{D_{\delta } D^{-1}}{g_j}.
\end{equation}
Multiplying the above number by $\ketbra{g_i}{g_j}$ and summing over $i,j$ from $1$ to $n$ gives
\begin{equation}
\label{ij3}
\delta X = \sum_{i,j=1}^{n} \pk{g_i}{D_{\delta } D^{-1}}{g_j} \dfrac{g_i^2 + g_j^2}{g_i + g_j} \ketbra{g_i}{g_j}.
\end{equation}
Let $\{\ket{k}\}_{k=1}^{n}$ represent the standard basis, then $\braket{k}{g_j}$ is complex number occuring in the $k$-th entry of $\ket{g_j}$. Using equations \eqref{Ddelta} and \eqref{DeltaX} we get $\pk{g_i}{D_{\delta } D^{-1}}{g_j} =\frac{1}{2} \sum_{l=1}^{n}  \ir{g_i}{l}{g_j}  \dfrac{(\delta X)_{ll}}{ d_{ll}(G)^2} $. The diagonal elements of $\delta X$ are then given by 
\begin{equation}
\begin{split}
\label{ij4}
(\delta X)_{kk} = \; & \sum_{l=1}^{n} \left( \dfrac{1}{2} \sum_{i,j=1}^{n} \is{k}{g_i}{g_j}{k} \dfrac{g_i^2 + g_j^2}{g_i + g_j}  \ir{g_i}{l}{g_j} \right) \dfrac{(\delta X)_{ll}}{ \left(d_{ll}(G)\right)^2} \\
 =  \; &  \sum_{l=1}^{n} \dfrac{1}{2} \left( O \Lambda O^\dag \right)_{kl}\dfrac{(\delta X)_{ll}}{ \left(d_{ll}(G)\right)^2},
\end{split}
\end{equation}

where $O$ is an $n \times n^2$ matrix with matrix elements given by $O_{k,ij} = \is{k}{g_i}{g_j}{k}$, $\Lambda$ is an $n^2 \times n^2$ diagonal matrix with matrix elements $\Lambda_{ij,kl} =  \delta_{ik}\delta_{jl} \dfrac{g_i^2 + g_j^2}{g_i + g_j}$. It is easy to check that rows of $O$ are orthogonal. Since $\Lambda > 0$ and $O$ is of rank $n$, the matrix $\dfrac{1}{2} O \Lambda O^\dag$ positive definite. \medbreak Consider 

\begin{equation}
\label{Ddelta1}
\ket{D_{\delta X }} \equiv \begin{pmatrix}
     (\delta X)_{11}\\                                                                                                                                                                                                                                                                                                                                                                                                                                                                                                                                                                                                                                  
     (\delta X)_{22}\\                                                                                                                                                                                                                                                                                                                                                                                                                                                                                                                                                                                                                                  
     \vdots \\
     (\delta X)_{nn}\\                                                                                                                                                                                                                                                                                                                                                                                                                                                                                                                                                                                                                                  
     \end{pmatrix}.
\end{equation}
Then equation \eqref{ij4} can be rewritten as 
\begin{align}
\label{ij5}
& \left(\mathbb{1} - \dfrac{1}{2} O \Lambda O^\dag  D^{-2} \right) \ket{D_{\delta X }} = 0 \notag \\ 
\Longrightarrow  \; & \left(D^2 - \dfrac{1}{2} O \Lambda O^\dag \right) D^{-2} \ket{D_{\delta X }} = 0 
\end{align}
Let $\Lambda'$ be an $n^2 \times n^2$ diagonal matrix whose matrix elements are given by $\Lambda'_{ij,kl} =  \delta_{ik}\delta_{jl} \dfrac{2 g_i g_j}{g_i + g_j}$. Since $\Lambda' > 0$, $\dfrac{1}{2} O \Lambda' O^\dag$ is positive definite. After some amount of tedious algebra we find that the following equation holds true.
\begin{align} \label{ij6}
D^2 = \dfrac{1}{2} O \left( \Lambda + \Lambda' \right) O^\dag.                                                                                                                                                                                                                                                                                                                                                                                                                                                                                                           \end{align}
Hence $D^2 - \dfrac{1}{2}O\Lambda O^\dag$ = $\dfrac{1}{2} O \Lambda' O^\dag >0$. This implies that for equation \eqref{ij5} to be true $\ket{D_{\delta X }} = 0$. This implies (see equation \eqref{Ddelta1}) $(\delta X)_{ii}= 0$, which implies that $ 2 d_{ii}(G)\delta x_{ii} = 0$, which implies that $x_{ii}=0$, i.e., $D_\delta = 0$. Substituting $D_{\delta } = 0$ in equation \eqref{ij3} gives $\delta X = 0$. 

Hence, demanding $J_G\left( \delta X \right)=0 $ leads to the conclusion that $\delta X =0$. 
\end{subequations}
\end{enumerate}
This means that $J_G$ is non-singular, which then implies that $ Det\left(J_{G}\right) \neq 0 $. This proves the theorem.
\end{proof}

Theorem \ref{Jacthm} implies that IFT holds true for all $G_0 \in \mathcal{G}$, i.e., for all $G_0 \in \mathcal{G}$ one can define these $\phi_{ij}$ functions so that the points 1., 2. and 3. mentioned in IFT are satisfied. The third point in IFT, i.e., equation \eqref{impff1}, tells us that for any $G \in I_{G_0}$, $F(G, X) = 0$, when $x_{ij} = \phi_{ij}(G)$, $\forall \; 1 \leq i ,j \le n$. This is equivalent to stating that if one obtains the $\phi_{ij}$ functions, defined in some open neighbourhood $I_{G_0}$ of $G_0$, then $x_{ij}=\phi_{ij}(G), \; \forall \; 1 \le i,j \le n$, gives us \emph{some} solution for $X$ in equation \eqref{EOY} for any $G \in I_{G_0}$. If it is true that assigning $x_{ij} = \phi_{ij}(G), \; \forall \; 1 \le i,j \le n$, implies that $X > 0$, then obtaining the $\phi_{ij}$ functions in some open neighbourhood $I_{G_0}$ of $G_0$ gives us \emph{the} solution for MED of all $G \in I_{G_0}$.

\begin{theorem}
\label{Xge0}
When $G \in I_{G_0}$ and $x_{ij} = \phi_{ij}(G)$, $\forall \; 1 \leq i ,j \le n$, then $X>0$. 
\end{theorem}

\begin{proof}
Suppose not. 

Let there be some $G_1 \in I_{G_0}$ such that when $x_{ij}= \phi_{ij}(G_1)$, $\forall$ $1 \le i,j \le n$, then $X$ has some non-positive eigenvalues. 

Let $G(t) \equiv (1-t)G_0 + tG_1$ be a linear trajectory in $\mathcal{G}$. $G(t)$ starts from $G_{0}$ when $t=0$ and ends at $G_1$ when $t=1$. Note that eigenvalues of $X$ are continuous functions of $x_{ij}$, and when restricting $x_{ij}$ to be such that $x_{ij} = \phi_{ij}(G)$, $\forall \; 1 \le i,j \le n$, then $x_{ij}$ are continuous functions over $I_{G_0}$. Thus the eigenvalues of $X$ are continuous over $I_{G_0}$, whenever $x_{ij} = \phi_{ij}(G)$. 

This implies the following.

\begin{enumerate}
\item[(i.)] When $x_{ij} = \phi_{ij}(G(0))$ $\forall \; 1 \leq i, j \le n$, all eigenvalues of $X$ are positive. 
\item[(ii.)] When $x_{ij} = \phi_{ij}(G(1))$ $\forall \; 1 \leq i, j \le n$, some eigenvalues of $X$ are non-positive. 
\end{enumerate}

The intermediate value theorem tells us that since $\phi_{ij}$'s are continuous over $I_{G_0}$, (i.) and (ii.) imply that there must be some $t' \in (0,1]$, such that
\begin{enumerate}                                                                                                                                                                 \item[(i.)] $X>0$, when $x_{ij}=\phi_{ij}(G(t))$, for all $t \in [0, t')$, 
\item[(ii.)] for all $t \in ( t',1]$, $X$ is not necessarily positive definite, when $x_{ij}=\phi_{ij}(G(t))$, and finally
\item[(iii.)] $X$ has some $0$ eigenvalue(s) when $x_{ij}=\phi_{ij}(G(t'))$, i.e., when $t=t'$.
\end{enumerate}

When $t<t'$ then $X>0$, which also implies that $X = \mathcal{X}(G(t))$ holds true for the interval $t \in [0,t')$. Since $ \frac{\left(\mathcal{X}(G)\right)^2}{Tr\left( \left(\mathcal{X}(G)\right)^2\right)}  = \mathscr{R}_\mathcal{G}(G)$ \footnote{See equation \eqref{GXR} in the proof of theorem \ref{thmret} in section \ref{SSGMM}.} we get that $\frac{X^2}{Tr\left(X^2\right)}  =  \mathscr{R}_\mathcal{G}\left(G(t)\right)$, when $t<t'$. Since $\mathscr{R}_\mathcal{G}$ is continuous on $\mathcal{G}$ \footnote{See description below definition \ref{defG} in section \ref{intro}.} and eigenvalues of $X$ are continuous in $I_{G_0}$, it follows that when $t=t'$, $\frac{ X^2}{Tr\left(X^2\right)} =  \mathscr{R}_\mathcal{G}\left(G(t')\right)$. From (iii.) above it is seen that when $t=t'$, $X$ is singular; this implies that $\mathscr{R}_\mathcal{G}\left(G(t')\right)$ is singular as well, which is a contradiction since we know that $\mathscr{R}_\mathcal{G}$ is a function from $\mathcal{G}$ to $\mathcal{G}$ and all gram 
matrices in $\mathcal{G}$ are 
positive definite. 

This contradiction arose from the assumption that when $x_{ij} = \phi_{ij}(G_1)$, $X$ is not positive definite. This proves the theorem.
\end{proof}

Theorem \ref{Xge0} tells us that for any starting point $G_0 \in \mathcal{G}$, if we take any point $G \in I_{G_0}$, the $\phi_{ij}$'s obey the equality: $\phi_{ij}(G) = d_{ij}(G)$, $\forall \; 1 \leq i, j \le n$. Given this fact, from here onwards, we will represent the implicit functional dependence $\phi_{ij}$ by $d_{ij}$ itself. 

We can make a stronger statement about the behaviour of the functions $d_{ij}$ on $\mathcal{G}$. It is easier to do so if we define trajectories, like the one defined in the proof of theorem \ref{Xge0} in $\mathcal{G}$, and prove results about the behaviour of the $d_{ij}$'s with respect to the independent variable $t$. For that purpose, let $G_0, G_1 \in \mathcal{G}$ be distinct; define a linear trajectory in $\mathcal{G}$ from $G_0$ to $G_1$, $G:[0,1]\longrightarrow \mathcal{G}$ as \begin{equation} \label{Gt} G(t) = (1-t) G_0 + t G_1.\end{equation} We now apply the implicit function theorem to $F(G(t), X)$, where $X$ represents the variables whose implicit dependence we seek and $t$ is the independent variable. 

The analytic implicit function theorem \cite{imp} tells us that if $f_{ij}\left(G(t),\vv{x}\right)$ are analytic functions of the variables $t$ and $x_{kl}$, then $\phi_{kl}(G(t))$ (which are equal to $ d_{kl}(G(t))$) should also be analytic functions of the variable $t$ $\in [0,1]$. Equations \eqref{fij}, \eqref{fji} and \eqref{fii} tell us that $f_{ij}(G(t),\vv{x})$ are multivariate polynomials in the variables $t$ and $x_{kl}$, which implies that the $f_{ij}$'s are analytic functions of $t$ and $x_{kl}$. Thus $d_{kl}(G(t))$ are analytic functions of the variable $t$. This implies that, more generally, $d_{kl}$ are analytic functions over $\mathcal{G}$.

\subsection{Taylor Series and Analytic Continuation} 
\label{TayAna}

The fact that $d_{kl}$ are analytic functions on $\mathcal{G}$ allows us to Taylor expand $d_{kl}$ from some point in $\mathcal{G}$ to another point. Let us assume that we want to find the solution for MED of some gram matrix $G_1 \in \mathcal{G}$, and that we know the solution for MED of $G_0 \in \mathcal{G}$. Then we define a trajectory as was done in equation \eqref{Gt}. We will now show that using equation \eqref{calEOY} we can obtain the derivatives of $d_{kl}(G(t))$, upto any order, with respect to $t$; this allows us to Taylor expand the $d_{kl}(G(t))$ function about the point $t=0$. Analytically continuing from $t=0$ to $t=1$ allows us to obtain the values of $d_{kl}(G(t))$ at $t=1$.

First we show how to obtain the first order derivatives of $d_{kl}(G(.))$ with respect to $t$. We will abbreviate $D(G(t))$ $=$ $(d_{11}(G(t)),d_{22}(G(t)),\cdots,d_{nn}(G(t)))$ as $D(t)$ for convenience. Similarly $\mathcal{X}(G(t))$ will be abbreviated as $\mathcal{X}(t)$. It will be useful to denote separately the matrix  of off-diagonal elements of $\mathcal{X}(t)$. Thus define $N(t) \equiv \mathcal{X}(t) - (D(t))^2$. Equation \eqref{calEOY} can be re-written as $\left(D(t)^2+N(t)\right)^2 =D(t)G(t)D(t)$. Let $\Delta \equiv \frac{d G(t)}{dt} = G_1 - G_0$. Taking the total first order derivative on both sides of equation \eqref{calEOY} gives

\begin{equation}
\label{1order}
\begin{array}{c} 
\left( D(t)^2 + N(t) \right) \left( 2 D(t) \dfrac{d D(t)}{dt} + \dfrac{d N(t)}{dt} \right)+ \left( 2 D(t) \dfrac{d D(t)}{dt} + \dfrac{d N(t)}{dt} \right) \left( D(t)^2 + N(t) \right) \\   -  \; (D(t)G(t))\dfrac{d D(t) }{dt} -  \dfrac{d D(t) }{dt}G(t)D(t)  
\end{array} =  D(t) \Delta D(t),
\end{equation} where \begin{enumerate}
                      \item[] $\dfrac{d D(t)}{dt} = \left(\dfrac{d \; \left( d_{11}(t) \right)}{dt},\dfrac{d \; \left( d_{22}(t) \right)}{dt},\cdots,\dfrac{d \; \left( d_{nn}(t) \right)}{dt}\right)$,
                      \item[] $\left(\dfrac{dN(t)}{dt}\right)_{kl}= \dfrac{d}{dt} \left( d_{kl}(t) + i d_{lk}(t)\right)$ (when $k<l$),
                      \item[] $\left(\dfrac{dN(t)}{dt}\right)_{ii}= 0$, and 
                      \item[] $\left(\dfrac{dN(t)}{dt}\right)_{kl}= \dfrac{d}{dt}\left(d_{lk}(t) - i d_{kl}(t)\right)$ (when $k>l$).
                     \end{enumerate}
Thus we get $n^2$ coupled ordinary differential equations.  By substituting the values of $d_{kl}(0)$ in equation \eqref{1order} one can solve for $\frac{d \; d_{kl}(t)}{dt}{|}_{t=0}$. \smallbreak

The second order derivatives can be obtained in a similar fashion: taking the total derivative of LHS and RHS of the equation \eqref{1order} with respect to $t$ (i.e. the second order derivative of the LHS of equation \eqref{calEOY}) we get a set of $n^2$ coupled second order differential equations. Setting $t=0$ and using the values of $d_{kl}(0)$ and $\frac{d \; d_{kl}(t)}{dt}{|}_{t=0}$, one can solve the resulting (linear) equations to obtain the values of the unknowns $\frac{d^2 \; d_{kl}(t)}{dt^2}{|}_{t=0}$. \smallbreak

Continuing in this manner one can obtain the values of the derivatives of $d_{kl}(t)$ upto any order, at the point $t=0$. In the following equation we give the $k$-th order derivative of equation \eqref{calEOY} for this purpose.

\begin{equation}
\begin{split}
\label{1norder}
& \Big( D(t)^2 + N(t) \Big)\Big( 2 D(t) \dfrac{d^k D(t)}{dt^k} + \dfrac{d^k N(t)}{dt^k} \Big)+\Big( 2 D(t) \dfrac{d^k D(t)}{dt^k} + \dfrac{d^k N(t)}{dt^k} \Big) \Big( D(t)^2 + N(t) \Big) \\
 & - (D(t)G(t))\dfrac{d^k D(t) }{dt^k} -  \dfrac{d^k D(t) }{dt^k}G(t)D(t)\\ 
= & -\Bigg( \Big(D(t)^2+N(t)\Big)\Big( \sum_{l_1=1}^{k-1} \binom{k}{l_1} \big( (\frac{d}{dt} )^{l_1} D(t) \big) \big( (\frac{d}{dt} )^{k-l_1} D(t) \big)    \Big )  + \text{h.c.} \Bigg) \\
 & - \sum_{l_2=1}^{k-1} \binom{k}{l_2} \Big( (\frac{d}{dt})^{k_1} (D(t)^2 + N(t)) \Big)\Big( (\frac{d}{dt})^{k-k_1} (D(t)^2 + N(t)) \Big) \\
 & + \sum_{m_1=1}^{k-1} \binom{k}{m_1} \Big( (\frac{d}{dt})^{m_1} D(t) \Big)G(t)\Big( (\frac{d}{dt})^{k-m_1} D(t) \Big) \\
 & + k \sum_{m_2=0}^{k-1} \binom{k}{m_2} \Big( (\frac{d}{dt})^{m_2} D(t) \Big) \Delta \Big( (\frac{d}{dt})^{k-m_2} D(t) \Big).
\end{split}
\end{equation}

By substituting the values of all derivatives at $t=0$, one can expand the $d_{kl}$ functions about the point $t=0$. Analytic continuation is straightforward: by using the aforementioned Taylor expansion about $t=0$, one obtains the value of $d_{kl}$ at some $t = \delta t > 0$; one can then use the aforementioned method to obtain the values of derivatives of $d_{kl}$ at $t=\delta t$ and Taylor expand the $d_{kl}$ functions from $\delta t$ to $t > \delta t$. In this manner one can Taylor expand and analytically continue $d_{kl}$'s from $t=0$ to the point $t=1$.

The need for analytic continuation raises the following question: what is the radius of convergence for the Taylor series about some point $t$ in the interval $[0,1]$? The LHS of equations \eqref{1order} and \eqref{1norder} gives us a hint: $\frac{d^k D(t)}{dt^k}$ and $\frac{d^k N(t)}{dt^k}$ scale proportionally to the $k$-th power of $\Delta$ i.e.,

\begin{equation}
\begin{split}
\label{scaling}
\Delta   \longrightarrow \nu \Delta \Longrightarrow  \left( \dfrac{d^k D(t)}{dt^k},\dfrac{d^k N(t)}{dt^k}\right)   \rightarrow \left( \nu^k \dfrac{d^k D(t)}{dt^k}, \nu^k \dfrac{d^k N(t)}{dt^k}\right).
\end{split}
\end{equation}
This tells us that we need to keep $|| \Delta ||_2$ small to ensure that either $G_1$ falls within the radius of convergence of the $d_{kl}$ functions when expanded about the point $G_0$ or the number of times one is required to analytically continue from $t=0$ to $t=1$ is low. It is very difficult to obtain the exact radius of convergence for every point in $\mathcal{G}$ since the value of the radius of convergence differs for different points in $\mathcal{G}$\footnote{Particularly as one gets closer to points near the boundary of $\mathcal{G}$ (which lies outside $\mathcal{G}$), the radius of convergence becomes smaller.}.

For a given $G_1$ for which we wish to find the solution, it is desirable to find a $G_0$ so that $G_1$ falls within the radius of convergence of the $d_{kl}$ functions, when expanded about $G_0$. In the following we give a method to find such a $G_0$ for a given $G_1$. 

\subsubsection{Starting point which generally doesn't require analytic continuation}
\label{sub1}

Let $G_0 \in \mathcal{G}$ be some gram matrix with the property that the diagonal of the positive definite square root of $G_0$, i.e., $G_0^\frac{1}{2}$ has the property $\left(G_0^\frac{1}{2}\right)_{11}$ $=$ $\left(G_0^\frac{1}{2}\right)_{22}$ $=$ $\cdots$ $=$ $\left(G_0^\frac{1}{2}\right)_{nn}$. Substituting  $G^\frac{1}{2} = G_0^\frac{1}{2}$, along with $W = \mathbb{1}_n$ and $D = \left(G_0^\frac{1}{2}\right)_{11} \mathbb{1}_n$ in the LHS of equation \eqref{DGWhermit} gives us the RHS of equation \eqref{DGWhermit}, i.e., they all satisfy equation \eqref{DGWhermit}. It is also seen that when $D = \left(G_0^\frac{1}{2}\right)_{11} \mathbb{1}_n$, then $X = D G_0^\frac{1}{2}$ is a solution for the equation \eqref{EOY} for $G = G_0$, and since $D$ is a multiple of $\mathbb{1}_n$, $X >0$. Thus, when the diagonal of $G_0^\frac{1}{2}$ is a multiple of $\mathbb{1}_n$, the solution for the MED of corresponding gram matrix $G_0$ is known\footnote{ This result is well known. It corresponds to those cases when $\
mathscr{R}\left(\widetilde{P}\right) = \widetilde{P}$.}.

Thus for a given $G_1$, we want to find a starting point $G_0$ such that the diagonal elements of $G_0^{\frac{1}{2}}$ are all equal. For this purpose expand the positive square root of $G_1$ i.e., $G_1^{\frac{1}{2}}$ in an ONB of $\mathcal{H}_n$, which comprises of $\frac{\mathbb{1}}{\sqrt{n}}$ and the generalized Gell-Man matrices $\frac{\sigma_{lk}}{\sqrt{2}}$ where $ 1 \leq l,k \leq n $\cite{bloch}. Here the $ \sigma_{lk}$ matrices are defined as
\begin{equation}
\label{gell}
\sigma_{lk} = \left\{ \begin{array}{lr}
		 \ketbra{l}{k} + \ketbra{k}{l},  \; \text{ when } l<k, \\  \\
		 i\ketbra{l}{k} - i \ketbra{k}{l}),  \; \text{ when } l>k, \\ \\
		 \sqrt{ \frac{2}{l (l+1) } } \left( \sum_{j=1}^{l} \ketbra{j}{j} - l \ketbra{l+1}{l+1} \right)\delta_{lk}, \text{ when } 1 \leq l \leq n-1 
\end{array} \right.
\end{equation}
All generalized Gell-Mann matrices in equation \eqref{gell} have Hilbert-Schmidt norm $\sqrt{2}$. Let $G_1^\frac{1}{2}$ have the following expansion in these Gell-Mann matrices.

\begin{equation}
\label{G1exp}
G_1^\frac{1}{2} =  \gamma \dfrac{\mathbb{1}}{\sqrt{n}} + \sum_{j=1}^{n-1} \beta_j \dfrac{\sigma_{jj}}{\sqrt{2}} +  \sum_{\underset{l \neq k}{ l , k=1}}^{n} \zeta_{lk} \dfrac{\sigma_{lk}}{\sqrt{2}},
\end{equation}

where $\zeta_{lk}, \beta_j, \gamma$ are real numbers. Note that $\gamma>0$ s $G_1^\frac{1}{2}$. Based upon this define $G_0^{\frac{1}{2}}$ as

\begin{equation}
\label{G0exp}
G_0^\frac{1}{2} = \kappa \dfrac{\mathbb{1}}{\sqrt{n}} +  \sum_{l \neq k} \zeta_{lk} \dfrac{\sigma_{lk}}{\sqrt{2}},
\end{equation} where $\kappa =  \sqrt{ \gamma^2 + \sum_{j} \beta_j^2}$.

It is easily verified that $Tr(G_0) =1$. One needs to check if $G_0^\frac{1}{2} >0$ or not. Generally, it is true that $G_0^\frac{1}{2} > 0$. But if some eigenvalues of $G_1$ are close to $0$, this \emph{may} not hold. Suppose it holds (as is generally the case), $d_{ii}(0) = \kappa, \; d_{kl}(0)=Re \left( \left(G_0^\frac{1}{2} \right)_{kl}\right), \text{ when } k<l, \; d_{kl}(0)=Im\left( \left(G_0^\frac{1}{2} \right)_{lk} \right), \text{ when } k>l$.  $G_1$ generally falls within the radius of convergence of all $d_{kl}$ functions about the starting point $G_0$. In such circumstances one doesn't need analytic continuation; one can straightforwardly calculate $d_{kl}(1)$ from the Taylor series about $t=0$. If $G_0^\frac{1}{2}$, obtained this way, isn't positive definite, then this method fails and one needs another starting point.

\subsubsection{Starting points which generally require analytic continuation}
\label{sub2}

Another possible starting point is an ensemble of equiprobable orthogonal states; this ensemble's gram matrix is $G_0 = \frac{\mathbb{1}}{n}$ where $d_{ij}(0)=\delta_{ij}\frac{1}{\sqrt{n}}$. To \emph{drag} the solution from $G_0$ to $G_1$ one needs to divide the $[0,1]$ interval into subintervals and analytically continue the $d_{kl}$'s from the starting point of each subinterval to its corresponding ending point. Here it needs to be ensured that one doesn't overshoot beyond the radius of convergence of any of the $d_{kl}$ functions at the starting point of each subinterval. For this purpose it was found that it generally suffices to divide $[0,1]$ into \large$\lceil$\normalsize$  n^2 || \Delta  ||_2 $\large$\rceil$\normalsize subintervals. Generally the smaller the intervals, the lower the value of error. 

\subsubsection{Error-Estimation}
\label{ex1}

There is a simple method to estimate the degree of error in the process; this is based on the fact that when the solution, i.e., $d_{kl}(1)$'s are substituted in the LHS of equation \eqref{EOY} one should obtain the zero matrix, which isn't what we get due to errors. Thus the value of the Hilbert-Schmidt norm of the quantity on the LHS, i.e., the value of $|| (D(1)^2 + N(1))^2 - D(1)G(1)D(1))||_2$ gives us an estimate of the degree of error which has accumulated into the solution. The closer $|| \left(D(1)^2 + N(1)\right)^2 - D(1)G(1)D(1))||_2$ is to $0$, the lower the error. Note that \emph{one cannot decrease the error significantly by increasing the order upto which the Taylor series is expanded beyond a order of expansion. On the other hand error rates can be substantially reduced by decreasing the size of the subintervals}.

Thus having solved for $d_{kl}(1)$ with a high degree of accuracy, one can now obtain the optimal POVM. In the following we present an example for $n=5$. Note that while the precision  of the starting point is upto $20$ significant digits, only the first $6$ significant digits have been displayed. For lack of space sometimes quantities have been displayed with upto only $4$ significant digits.  \medbreak

\begin{tabular}{l l} 
$\tket{\psi}{1} = \left[\begin{array}{c} 
                               0.320457  \\ 0.123687 + i0.0117558  \\  0.117838+ i-0.027942   \\   0.109674+ i0.0167151   \\   0.0860555+ i0.00780123     
                              \end{array}\right]$ & 
$\tket{\psi}{2} = \left[\begin{array}{c} 
                              0.123687 -i0.0117558  \\ 0.397851 \\  0.169692 -i0.0506685 \\   0.125198 -i0.0244774 \\   0.124106 -i0.0261114 
                              \end{array}\right]$ \\ ~ & ~ \\
$\tket{\psi}{3} = \left[\begin{array}{c} 
                              0.117838 + i0.027942  \\ 0.169692 + i0.0506685  \\  0.404725  \\   0.13847+ i0.0177653   \\   0.122277-i0.0249506 
                              \end{array}\right]$ &
$\tket{\psi}{4} = \left[\begin{array}{c} 
                              0.109674  -i0.0167151  \\ 0.125198 + i0.0244774  \\  0.13847-i0.0177653 \\   0.373791 \\   0.110387 -i0.013984  
                              \end{array}\right]$ \\ ~ & ~ \\
$\tket{\psi}{5} = \left[\begin{array}{c} 
0.0860555 -i0.00780123  \\ 0.124106 + i0.0261114  \\  0.122277+ 
i0.0249506   \\   0.110387+ i0.013984   \\   0.33677                        
                              \end{array}\right] $. \end{tabular}  
\smallskip

The corresponding $\tket{u}{i}$ states are given by: \smallskip

\begin{tabular}{l l}
$\tket{u}{1} = \left[\begin{array}{c} 
                               3.93887   \\ -0.668108 + i0.0313699  \\  -0.553671 + i0.331697  \\   -0.611991 - i0.234777   \\   -0.375925 - i0.264517    
                              \end{array}\right]$ & 
$\tket{u}{2} = \left[\begin{array}{c} 
                             -0.668108 - i0.0313699  \\ 3.52494 \\  -0.939093 + i0.353801   \\   -0.418308 + i0.204928   \\   -0.685643 + i0.0142212     
                              \end{array}\right]$ \\ ~ & ~ \\
$\tket{u}{3} = \left[\begin{array}{c} 
                              -0.553671 - i0.331697  \\ -0.939093 - i0.353801  \\  3.50731   \\ -0.634577 - i0.0903887   \\   -0.554402 + i0.418281  
                              \end{array}\right]$ &
$\tket{u}{4} = \left[\begin{array}{c} 
                              -0.611991 + i0.234777  \\ -0.418308 - i0.204928  \\  -0.634577 + i0.0903887   \\   3.42828   \\   -0.568152 + i0.0597921     
                              \end{array}\right]$ \\ ~ & ~ \\
$\tket{u}{5} = \left[\begin{array}{c} 
-0.375925 + i0.264517  \\ -0.685643 - i0.0142212  \\  -0.554402 - i0.418281   \\   -0.568152 -i0.0597921   \\   3.74634                         
                              \end{array}\right] $.\end{tabular}   
\medbreak

The gram matrix for the ensemble $\{ \tket{\psi}{i} \}_{i=1}^{5}$, i.e., $G_1$ is given by:
\footnotesize
\medbreak
\begin{align*}
&G_1 = \\
&\left[ \begin{array}{lllll}
0.15257   & 0.13405  -i0.017665  &  0.13285+ i0.021068   &   0.11811-i0.010337   &   0.098267+ i0.0026888   \\  
0.13405 + i0.017665  & 0.23744   &  0.18316+ i0.051216   &   0.14883+ i0.02325   &   0.13487+ i0.034111   \\  
0.13285 -i0.021068  & 0.18316 -i0.051216  &  0.24489   &   0.15659-i0.020010   &   0.13850+ i0.013294   \\  
0.11811 + i0.010337  & 0.14883 -i0.023258  &  0.15659+ i0.020010   &   0.20017   &   0.12067+ i0.016377   \\  
0.098267 -i0.0026888  & 0.13487 -i0.034111  &  0.13850-i0.013294   &   0.12067-i0.016377   &   0.16492  
\end{array}\right].
\end{align*}
\normalsize
\medbreak
Then using equation \eqref{G0exp}, we have
\footnotesize
\begin{align*}
&G_0^\frac{1}{2} = \\
&\left[ \begin{array}{lllll}
0.36821  & 0.12368  -i0.011755  &  0.11783+ i0.02794   &   0.10967-i0.016715   &   0.086055-i0.0078012   \\  
0.12368 + i0.011755  & 0.36821 &  0.16969+ i0.050668   &   0.12519+ i0.024477   &   0.12410+ i0.026111   \\  
0.11783 -i0.02794  & 0.16969-i0.050668  &  0.36821   &   0.13847-i0.017765   &   0.12227+ i0.024950   \\
0.10967 + i0.016715  & 0.12519 -i0.024477  &  0.13847+ i0.017765 &   0.36821   &   0.11038+ i0.013984   \\ 
0.086055 + i0.0078012  & 0.12410 -i0.026111  &  0.12227-i0.024950   &   0.11038-i0.01398   &   0.36821     
\end{array}\right].
\end{align*}

\normalsize
\medbreak
We see that all the diagonal elements of $G_0^\frac{1}{2}$ are all equal. Also $G_0^\frac{1}{2}>0$. Thus $d_{ii}(0)$ is equal to the diagonal elements of $G_0^\frac{1}{2}$ and $d_{kl}(0)$ are assigned values of the off-diagonal elements of ${G_0}^\frac{1}{2}$ (when $i \neq j)$. 

Here $|| \Delta ||_2 = ||G_1 - G_0||_2 =  0.058777 \sim 1/5^{2} \; ( \; =0.04 \;)$. This gives us the indication that $t=1$ lies within the radius of convergence of the implicitly defined functions $d_{kl}$ about the point $t=0$ and that no analytic continuation is required at any intermittent point. Upon employing the Taylor series expansion and expanding the series upto $10$-th term, we obtain the solution for $\mathcal{X}(1)= {D(1)}^2+N(1)$: 
\footnotesize
\begin{align*}
&\mathcal{X}(1)= {D(1)}^2+N(1) = \\
&\left[ \begin{array}{lllll}
0.09627  & 0.04197  -i0.00407  &  0.04054+ i0.009487   &   0.03528-i0.005896   &   0.02484-i0.003121   \\
0.04197 + i0.00407  & 0.1635 &  0.07237+ i0.02128   &   0.04981+ i0.009339   &   0.04439+ i0.008852   \\   
0.04054 -i0.009487  & 0.07237 -i0.02128  &  0.1710   &   0.05580 -i0.00729   &   0.04424+ i0.008926   \\
0.03528 + i0.005896  & 0.04981 -i0.009339  &  0.05580+ i0.00729   &   0.1399   &   0.03732+ i0.004563   \\ 
0.02484 + i0.003121  & 0.04439 + i-0.008852  &  0.04424-i0.008926   &   0.03732-i0.004563   &   0.1083
\end{array}\right].
\end{align*}

\normalsize
\medbreak

$\mathcal{X}(1)>0$ holds true.

$d_{11}(1) = 0.310278, \; d_{22}(1) = 0.404377, \; d_{33}(1)=0.413591, \; d_{44}(1)=0.374064, \; d_{55}(1)=0.329225$.

The maximum success probability, $P_s^{max}= \sum_{i=1}^{n} (d_{ii}(1))^2 = 0.679164$.

$|| (\mathcal{X}(1))^2 - D(1)G(1)D(1)||_2 =  2.92337 \times 10^{-9}$.

For lack of space the projectors $\ketbra{w_i}{w_i}$ aren't given here. Instead we give the ONB $\{ \ket{w_i} \}_{i=1}^{n}$: 

\medbreak
\begin{tabular}{l l}
$\ket{w_1} = \left[ \begin{array}{c} 0.998902 -i 0.000902941 \\ -0.0294294 -i0.00140465  \\  -0.0281464+ i0.0114238   \\   -0.0185558-i 0.00595048   \\  -0.00450716-i0.00157192      \end{array} \right]$ & 
$\ket{w_2} = \left[ \begin{array}{c} 0.0295208 -i0.00161874  \\ 0.999231 -i0.000890303  \\  -0.00479151+ i0.00107801   \\   0.00760396-i0.00239694   \\   0.0239121-i0.00195944        \end{array} \right]$ \\ ~ & ~ \\
$\ket{w_3} = \left[ \begin{array}{c} 0.0285947 + i0.0113073  \\ 0.00328547 + i0.000850588  \\  0.999104+ i0.000230941   \\ 0.0125773+ i0.00210581   \\   0.0237566-i0.0103508    \end{array} \right]$ &
$\ket{w_4} = \left[ \begin{array}{c} 0.0179661 -i0.00612077  \\ -0.00850615 -i0.00226113  \\  -0.0132936+ i0.00235778   \\  0.999616+ i0.00060647   \\   0.0121594-i0.000373086   \end{array} \right]$ \\ ~ & ~ \\
$\ket{w_5} = \left[ \begin{array}{c} 0.00301285 -i0.00208885 \\ -0.0240121 -i 0.00194482  \\ -0.0235693-i0.0103215  \\   -0.0127196-i0.000528417  \\   0.99929+ i0.000965318 \end{array} \right]$.\end{tabular} \medbreak

Despite having satisfied the rotationally invariant conditions (refer theorem \ref{thm2}), we would like to see if both the conditions \eqref{St} and \eqref{Glb} are satisfied. Instead of checking condition \eqref{St} we check if $Z$, from equation \eqref{Z}, is hermitian or not. We first use $\{ \ket{w_i} \}_{i=1}^{n}$ to  compute the operator $Z$. We measure the non-hermiticity of $Z$ as $\frac{1}{2}|| Z - Z^\dag ||_2$, which takes the value $2.22059\times 10^{-10}$ for our example. That $Z$  is hermitian (within error) and satisfies equation \eqref{Z} implies that equations \eqref{cslack} or equivalently equations \eqref{St} are satisfied. Additionally we find that $\forall \; 1 \leq i \leq n$, all except one eigenvalue of $Z - p_i \ketbra{\psi_i}{\psi_i}$ are positive. For each $i=1,2,\cdots,n$ the non-positive eigenvalue of $Z-p_i \ketbra{\psi_i}{\psi_i}$ is either $0$ or of the order $10^{-10}$, showing that the condition \eqref{Glb} is also satisfied. Thus we have demonstrated an example of obtaining 
the optimal POVM for MED of an ensemble of $5$ LI states. 

\subsection{Algorithms: Computational Complexity}
\label{algcompl}

In the following we outline the algorithm for the Taylor series expansion method, which gives us the solution for the MED of a given $n$-LIP ensemble. The method has already been elucidated in detail in subsubsection \ref{sub2}. After giving the algorithm, we give its time complexity\footnote{The time complexity of any algorithm is given by the order of the total number of elementary steps involved in completing said algorithm. In this paper, each of the following are regarded as elementary steps: basic arithmetic operations (addition, subtraction, multiplication, division) of floating point variables,assigning a value to a floating-point variable, checking a condition and retrieving the value of a variable stored in memory.} and space complexity\footnote{The space complexity is the count of the total number of variables and constants used in algorithms. These variables and constants can be of floating point type, integer type, binary etc; in this paper we treat them all alike while adding the number of 
variables to give us the final count. Similar to the case of the time complexity, space complexity too is given in terms of the \emph{order} of the count, rather than the exact number.}. The acceptable tolerance error being assumed here is of the order $10^{-9}$, and the time and space complexities are computed corresponding to this acceptable error margin.  \bigskip

\textbf{Algorithm 1: Taylor Series} The algorithm of the Taylor series method (subsubsection \ref{sub2}) is given in the following steps.

\begin{enumerate}
 \item[(1)] Construct the gram matrix $G_1$ from the given ensemble $\widetilde{P}$. Choose an appropriate starting point $G(0)=G_0$ (for which the solution $d_{ij}(0)$, for all $1 \le i,j \le n$, is known) and define the function $G(t)= (1-t)G_0 + t G_1$. If $|| \Delta ||_2 n^2 \sim 1$ then there's no need to divide the interval $[0,1]$ into subintervals, but otherwise divide $[0,1]$ into $L\equiv \left\lceil|| \Delta ||_2 n^{2}\right\rceil$ intervals.
\item[(2)] For $l=0,1,2,\cdots,L-1$, set $t_l= \dfrac{l}{L}$ and iterate over each interval in the following manner:
\begin{enumerate}
\item	[(2.1)] For $k = 1$ through $k = K$ iterate the following:
solve eqn \eqref{1norder} for $\dfrac{d^k d_{ij}}{dt^k}|_{t=t_l}$, for all $ 1 \le i, j \le n$, by using values of lower order derivatives as explained in subsection \ref{TayAna}.

\item[(2.2)] Having obtained the values of derivatives $\dfrac{d^{k} d_{ij}}{dt^k}|_{t=t_l}$ upto $K$-th order for all $1 \le i,j \le n$, substitute these derivatives in an expression of the Taylor series expansion of the $d_{ij}$ functions about the point $t=t_l$, when expanded to $K$-th order. The resulting expressions will give $K$-th degree polynomials in the variable $t$ for each $1 \le i, j \le n$, i.e, for each $d_{ij}$. Obtain the value of $d_{ij}(t_{l+1})$ by computing the value these polynomials take at $t=t_{l+1}$. Then increment $t$ to $t_{l+1}$, go to (2) and iterate. Stop when $l=L$. 
\end{enumerate}
\end{enumerate}

In the following table we give the time and space complexity of various steps in the aforementioned algorithm.\medbreak

\small
\begin{tabular}{ | p{1em} |p{25em}|p{7em}| p{7em} |}
 \hline
 \multicolumn{4}{|c|}{Computational Complexity for Taylor Series Method} \\
 \hline
 & Step in the algorithm     & \small{Time Complexity} & \small{Space Complexity} \\
 \hline
 1. & Computing $G_1$ from $\widetilde{P}$  & $\mathcal{O}(n^3)$    & $\mathcal{O}(n^2)$\\
 2. & Computing $G_0$ from $G_1$, as in subsubsection \ref{sub1}  & $\mathcal{O}(n^3)$    & $\mathcal{O}(n^3)$\\
 3. & Solving for $\frac{d^k \; d_{ij}}{dt^k}|_{t=t_l}$ for $k = 1,2,\cdots,K$  &   $\mathcal{O}(K n^6)$  & $\mathcal{O}(K n^4)$\\
 4. & Computing Taylor series expansion of $d_{ij}(t-t_l)$ upto $K$-th order at $t=t_{l+1}$ & $\mathcal{O}(Kn^2)$ & $\mathcal{O}(n^2)$ \\
 5. & Repeating steps 3. and 4. over $L \simeq n^2 || \Delta ||_2$ times & $\mathcal{O}(K^2 n^8)$ & $\mathcal{O}( n^6)$ \\
 \hline
\end{tabular} \medbreak
\normalsize

Note that the algorithm is polynomial in $n$. It is expected that to maintain the acceptable error margin (i.e., $||\left(D(1)^2+N(1)\right)^2$ $-$ $D(1)G(1)D(1) ||_2$ $\lesssim$ $10^{-9}$) as $n$ increases, one would have to increase the value of $K$ as well. While the numerical examples we checked support this hypothesis, the required increment of $K$ to compensate the increase in the value of $n$ is seen to be significant only over large variations of values of $n$ (when $n$ varies over a range of $20$). Indeed, it remains almost constant for $n = 3$ to $n=10$ for the error to remain within the margin of the order of $10^{-9}$. As in the example given in subsubsection \ref{ex1}, choosing $K=10$ suffices to maintain the error within said margin.

If $|| \Delta  ||_2 n^2 \simeq 1$, analytic continuation isn't required and then the total time complexity of the algorithm is $\mathcal{O}(n^6)$ and the total space complexity of the algorithm is $\mathcal{O}(n^4)$. In case $|| \Delta  ||_2 n^2 > 1$, since the maximum value of $|| \Delta ||_2 \le 2$, analytic continuation is required, and in that case, the worst case time and space complexities\footnote{That is, worst-case corresponding to the value of $|| \Delta ||_2$.} are given by $\mathcal{O}( n^8)$ and $\mathcal{O}( n^6)$ respectively.

While the Taylor series method is polynomial in time with a relatively low computational complexity, it is seen that directly employing Newton's method is simpler and more computationally efficient. We will now explain how to employ Newton's method. 

\textbf{Algorithm 2: Newton's Method} This is a well known numerical technique for solving non-linear equations. We use it here to solve the equations $f_{ij}(G, \vv{x}) =0$, $\forall$ $1 \le i,j \le n$,  (see  \eqref{fij}, \eqref{fji} and \eqref{fii} for $f_{ij}$) where $G$ is the gram matrix of the ensemble $\widetilde{P}$ whose MED we want to solve for, and $\vv{x}$ are the variables which - we demand - will converge to the solution $\vv{d}(G)$. This convergence is achieved over a few iterations which are part of the algorithm. The technique is based on a very simple principle which we will now elaborate. 

The Taylor expansion of the $f_{ij}(G,.)$ functions, when expanded about the point $\vv{d}(G)$, can be approximated by the first order terms for small perturbations $\vv{d}(G)$ $\longrightarrow$ $\vv{d}(G)+\delta \vv{x}$ as seen in the equation below.

\begin{equation}
\label{tay1}
\begin{split}
f_{ij}(G, \vv{d}(G) + \delta \vv{x}) &  \approx \; f_{ij}(G, \vv{d}(G)) +  \sum_{k,l=1}^{n} \left( \dfrac{\partial f_{ij}(G,\vv{x})  }{\partial x_{kl}  } |_{\vv{x}=\vv{d}(G)}\right) \delta x_{kl} \\
&  = \;  \sum_{k,l=1}^{n} \left(J_{G}\right)_{ij,kl} \delta x_{kl},
\end{split}
\end{equation}
where we have used $f_{ij}(G, \vv{d}(G)) =0, \; \forall \; 1 \le i,j \le n$, and where we denote $\left(J_{G}\right)_{ij,kl} \equiv \dfrac{\partial f_{ij}(G,\vv{x})  }{\partial x_{kl}  } |_{\vv{x}=\vv{d}(G)}$. 

We want to obtain the value of the solution $\vv{d}(G)$. We assume that our starting point is $\vv{d}(G)+\delta \vv{x}$ which is close to $\vv{d}(G)$, so that $f_{ij}(G,\vv{d}(G)+\delta \vv{x})$ can be approximated as the RHS of equation \eqref{tay1}. Denote the inverse of the Jacobian $J\left(G,\vv{d}(G)\right)$ by $\left(J_{G}\right)^{-1}$\footnote{In theorem \eqref{Jacthm} we proved that the Jacobian is non-singular, so we know that the inverse \emph{will} exist.}. Then we get

\begin{equation}
\label{tay2}
\sum_{k,l=1}^{n} \left( \left(J_{G}\right)^{-1} \right)_{ij,kl} f_{kl}(G, \vv{d}(G) + \delta \vv{x}) \approx \delta x_{ij}, \; \forall \; 1 \le i,j \le n.
\end{equation}
 
Subtracting $\delta \vv{x}$ from $\vv{d}(G)+\delta \vv{x}$ gives us $\vv{d}(G)$, which is the required solution. The catch here is that since we do not know the solution $\vv{d}(G)$ to start with, we cannot compute the Jacobian $J(G,\vv{d}(G))$. But since $\vv{d}(G) + \delta \vv{x}$ is close to $\vv{d}(G)$, we can approximate $J(G,\vv{d}(G))$ by $J(G,\vv{d}(G)+ \delta \vv{x})$, which we can compute. So we use $J(G,\vv{d}(G)+ \delta \vv{x})$ in place of $J(G,\vv{d}(G))$ in the algorithm, particularly, instead of using $\left( J_{G} \right)^{-1}$, computed at $\vv{d}(G)$, in equation \eqref{tay2}, we use it when computed at the point $\vv{d}(G) + \delta \vv{x}$.

The description of the principle behind Newton's method clarifies the algorithm, whose steps we list below. 

Starting with $x_{ij}^{(0)}=\dfrac{1}{\sqrt{n}} \delta_{ij}$ (for all  $1 \le i,j \le n$), $k=1$ and assuming $\epsilon = 10^{-9}$, iterate
\begin{enumerate}
\item[(1)] Substitute $x_{ij}^{(k-1)}$ into the functions $f_{ij}(G,.)$ defined in equations \eqref{fij}, \eqref{fji} and \eqref{fii}. Arrange all the $f_{ij}$'s in a single column, which will have $n^2$ rows; we will denote this $n^2$-row long column by $\gamma^{(k-1)}$. 
\item[(2)] Stop when $|| \gamma^{(k-1)} ||_2 < \epsilon$.
\begin{enumerate}
\item[(2.1)] Compute the Jacobian, $J_G^{(k-1)}$, where $\left(J_G^{(k-1)}\right)_{ij,st} = \dfrac{\partial f_{ij}(G,\vv{x}) }{\partial x_{st} } $ at the point $\vv{x} = \vv{x}^{(k-1)}$. 
\item[(2.2)] Compute the the inverse of $J_G^{(k-1)}$ i.e. $\left( J_G^{(k-1)} \right)^{-1}$.
\item[(2.3)] $x_{ij}^{(k)} =x_{ij}^{(k-1)} -  \left( \left( J_G^{(k-1)} \right)^{-1} \gamma^{(k-1)} \right)_{ij}$.
\end{enumerate}
\end{enumerate}

For each $n=3$ to $n=20$, we tested approximately twenty-thousand different examples, each of which for we obtained the required solution within the margin error. What's more, it was also seen that the maximum number of iterations required to maintain the error tolerance was constant over this range of $n$, more specifically, for each of these examples we required the number of iterations to be ten. Since the number of iterations required doesn't increase with $n$ (or increases very slowly), the computational complexity (time and space) of this algorithm is determined by the cost of steps within an iteration. Keeping this in mind, we give the computational complexity (time and space) of this algorithm in the following table. \medbreak

\small
\begin{tabular}{ | p{1em} |p{25em}|p{7em}| p{7em} |}
 \hline
 \multicolumn{4}{|c|}{Computational Complexity for Newton's Method} \\
 \hline
 & Step in the algorithm     & \small{Time Complexity} & \small{Space Complexity} \\
 \hline
 1. & Computing values of $f_{ij}(G,\vv{x}^{(k-1)})$, by substituting $\vv{x}^{(k-1)}$ into equations \eqref{fij},\eqref{fji} and \eqref{fii}, for all $ 1\le i,j \le n$  & $\mathcal{O}(n^2)$    & $\mathcal{O}(n^2)$\\
 
 2. & Computing the Jacobian, $J_G^{(k-1)}$ at the point $\vv{x}^{(k-1)}$  & $\mathcal{O}(n^4)$    & $\mathcal{O}(n^4)$\\
 
 3. & Computing the inverse of the Jacobian $\left(J_G^{(k-1)}\right)^{-1}$ from the Jacobian, at the point $\vv{x}^{(k-1)}$  &   $\mathcal{O}(n^6)$  & $\mathcal{O}(n^4)$\\
 
 4. & Computing $\vv{x}^{(k)}$ using $\left(J_G^{(k-1)}\right)^{-1}$ and $\vv{x}^{(k-1)}$ (point 2.3 in the list of steps of this algorithm above) & $\mathcal{O}(n^4)$ & $\mathcal{O}(n^2)$ \\
 
  \hline
\end{tabular}\normalsize \medbreak

Thus we see that the time complexity of Newton's method is $\mathcal{O}(n^6)$ and the space complexity is $\mathcal{O}(n^4)$. The number of steps involved are lower than in the Taylor series method, making this algorithm simpler, and also the computational complexity (both time and space) of Newton's method is lower than that of the Taylor series' method when one cannot find a close enough starting point $G_0$ to the given point $G_1$ in the latter method.

Let's compare the efficiency of these methods to that of an SDP algorithm. We will employ an SDP algorithm known as the Barrier-type Interior Point Method (IPM) \cite{Boyd}. 

\textbf{Algorithm 3: Barrier-type IPM (SDP)} The SDP problem corresponding to MED is given by \eqref{dual}. The objective of this problem is to minimize the value of $Tr(Z)$ subject to the constraints: $Z \ge p_i \ketbra{\psi_i}{\psi_i},$ $\forall$ $i=1,2,\cdots,n$.

In this method we obtain $Z$ which solves \eqref{dual} over a series of iterations, known as outer iterations. One starts the $k$-th such iteration with an input $Z^{(k-1)}$ - a candidate for $Z$ - and ends with an output $Z^{(k)}$, which will serve as the input for the next iteration. The $Z^{(k)}$, which are successive approximations for $Z$, take values within the feasible region, i.e, the region given by the set $\{ Z \text{ is $n \times n$, positive definite} \; | \;Z \ge p_i \ketbra{\psi_i}{\psi_i}$, $\forall$ $1 \le i \le n\}$. If $Z$ lies in the interior of this feasible region then it is such that $Z > p_i \ketbra{\psi_i}{\psi_i}$, $\forall$ $1 \le i \le n$, whereas if $Z$ is a boundary point of the feasible region then there is some $i = 1,2,\cdots, n$ such that $Z - p_i \ketbra{\psi_i}{\psi_i}$ has at least one zero eigenvalue.

In the first iteration, one starts with some \emph{strictly} feasible $Z=Z^{(0)}$, i.e., some $Z^{(0)}$ which lies in the interior of the feasible region. To ensure that $Z^{(k)}$ remains within the feasible region one perturbs the objective function which is being minimized: instead of performing an unconstrained minimization of $Tr(Z)$, one performs an unconstrained minimization of $Tr(Z) - \dfrac{1}{w}\sum_{i=1}^{n} Log(Det(Z-p_i\ketbra{\psi_i}{\psi_i}))$, where $\dfrac{1}{w}$ is a weight factor. The reason behind subtracting the expression $\dfrac{1}{w}$ $\sum_{i=1}^{n} Log(Det(Z-p_i \ketbra{\psi_i}{\psi_i}))$ from $Tr(Z)$ for unconstrained minimization, is that the expression $Log(Det(Z-p_i\ketbra{\psi_i}{\psi_i}))$ tends to infinity as $Z$ approaches the boundary of the feasible region. Thus performing unconstrained minimization of $Tr(Z)-\dfrac{1}{w }\sum_{i=1}^{n} Log(Det(Z-p_i\ketbra{\psi_i}{\psi_i}))$ will ensure that while the candidates for $Z$, viz, $Z^{(k)}$, may inch closer to 
the boundary of the feasible region, they will never cross it.

The unconstrained minimization of $Tr(Z) - \dfrac{1}{w}\sum_{i=1}^{n} Log(Det(Z-p_i\ketbra{\psi_i}{\psi_i}))$ is performed using Newton's method. The iterations within Newton's method are known as inner iterations. Newton's Method is performed as follows: using the generalized Gell-Mann basis, expand $Z = \sum_{i,j=1}^{n} y_{ij} \frac{\sigma_{ij}}{\sqrt{2}}$, where $\sigma_{nn} = \text{\small$\sqrt{\frac{2}{n}}$} \mathbb{1}_n$. Obtain the equations 
\begin{equation}
\begin{split}
\label{hij}
h_{kl}(\vv{y}) & \equiv \; \dfrac{\partial \; \left( Tr(Z) - \dfrac{1}{w}\sum_{i=1}^{n} Log(Det(Z-p_i\ketbra{\psi_i}{\psi_i})) \right)}{\partial \; y_{kl}} \\ 
= & \; \sqrt{n} \delta_{k,n}\delta_{l,n} - \dfrac{1}{w} \sum_{i=1}^{n} Tr \left(   \left( Z - p_i \ketbra{\psi_i}{\psi_i} \right)^{-1} \frac{\sigma_{kl}}{\sqrt{2}} \right).
\end{split}
\end{equation}

We want to solve for the equations $h_{ij}=0,$ $\forall$ $1 \le i,j \le n$ using Newton's method. The algorithm is the same as the one described above. These equations give the stationary points of the functions $h_{ij}$, $\forall \; 1 \le i, j \le n$. The matrix elements of the Jacobian of the $h_{ij}$ functions with respect to the $y_{kl}$ variables take the following form\footnote{This isn't difficult to derive; alternately the Barrier-type IPM algorithm for MED is also given in section 11.8.3 in \cite{Boyd} (p. 618), wherein expression for the matrix elements of the Jacobian has been explicitly given.}.
\begin{equation}
\begin{split}
\label{Jacboyd}
 H_{kl,st} & \equiv \; \dfrac{\partial \; h_{kl}(\vv{y})  }{\partial \; y_{st}} = \; \dfrac{\partial^2 \; \left(  Tr(Z) - \dfrac{1}{w}\sum_{i=1}^{n} Log(Det(Z-p_i\ketbra{\psi_i}{\psi_i}))  \right)}{\partial \; y_{kl} \partial \; y_{st}}    \\    = \; & \dfrac{1}{w} \sum_{i=1}^{n} Tr\left( \left(Z-p_i \ketbra{\psi_i}{\psi_i} \right)^{-1} \frac{\sigma_{kl}}{\sqrt{2}}\left(Z-p_i \ketbra{\psi_i}{\psi_i} \right)^{-1}\frac{\sigma_{st}}{\sqrt{2}} \right),
\end{split}
\end{equation}

where $H_{kl,st}$ are the matrix elements of the Jacobian, as can be seen from equation \eqref{Jacboyd}. Let $\vv{\alpha} \in \mathbb{C}^{n^2}$ be some non-zero complex $n^2$-tuple, and let $A \equiv \sum_{i,j=1}^{n} \alpha_{ij} \frac{\sigma_{ij}}{\sqrt{2}}$. Then we have the equality

\begin{align}
\label{Hve}
 &  \sum_{k,l,s,t=1}^n \alpha_{kl}^* H_{kl,st} \alpha_{st}  \\  = & \frac{1}{w}  \sum_{i=1}^{n} Tr\left( \left( (Z - p_i \ketbra{\psi_i}{\psi_i})^{-\frac{1}{2}} A^\dag (Z - p_i \ketbra{\psi_i}{\psi_i})^{-\frac{1}{2}}   \right)\left( (Z - p_i \ketbra{\psi_i}{\psi_i})^{-\frac{1}{2}} A (Z - p_i \ketbra{\psi_i}{\psi_i})^{-\frac{1}{2}}   \right)  \right)> 0.\notag
\end{align}

This inequality is true for all non-zero $\vv{\alpha} \in \mathbb{C}^{n^2}$. Thus the Jacobian $H$, whose matrix elements are given in equation \eqref{Jacboyd}, is positive definite throughout the feasible region. Thus the \emph{only} stationary points in the feasible region can be local minima. But since $H > 0 $ throughout the feasible region, there can only be \emph{one} local minima in said region, i.e., the stationary point gives \emph{the} minima which we are searching for\footnote{There is another way to appreciate this: since the function $Tr(Z) - \frac{1}{w} \sum_{i=1}^{n} Log \left( Det \left(Z-p_i \ketbra{\psi_i}{\psi_i} \right) \right)$ is a convex function over the feasible region, there can only be one minima in said region, which corresponds to the point we want. The convexity of the Log-Determinant function $-Log \left( Det(X) \right)$ over the space $\{\text{all $n \times n$ matrices } X | X \ge 0\}$ is established in section 3.1 on p. 73 in \cite{Boyd}.}. 

Thus the inner iterations give us the minima point $Z^{(k)} = \sum_{i,j=1}^{n} y_{ij}^{(k)} \frac{\sigma_{ij}}{\sqrt{2}}$  corresponding to some weight factor $\frac{1}{w^{(k-1)}}$. After having found the minima point $Z^{(k)}$ in the $k$-th iteraction, the $k+1$-th iteration is commenced with changing the weight of the barrier function, i.e., $w^{(k-1)}$ $\longrightarrow$ $w^{(k)}> w^{(k-1)}$, and performing an unconstrained minimization of $Tr(Z) - \dfrac{1}{w^{(k)} }\sum_{i=1}^{n} Log(Det(Z-p_i\ketbra{\psi_i}{\psi_i}))$, starting from the point $Z^{(k)}$. These iterations are continued until the weight of the barrier function decreases to an insignificantly small number (i.e. given by the error tolerance). The final approximation $Z^{(k_{f})}$ is then declared as the solution.

We briefly outline the steps in the algorithm below. \bigskip

Let $\epsilon$ be the error tolerance for the algorithm. For starting the algorithm choose the following: the value of $\mu$ between $\sim 3$ to $100$, the weight $w^{(0)}$  $\sim$ $10$, the initial starting point for $Z$ as $Z^{(0)}=\mathbb{1}_n$, then set $k=1$ and iterate the following.

\begin{itemize}
\item[(1)] Perform unconstrained minimization of the function $Tr(Z) - \dfrac{1}{w^{(k-1)} }\sum_{i=1}^{n} Log(Det(Z-p_i\ketbra{\psi_i}{\psi_i}))$ with starting point as $Z = Z^{(k-1)}$ (using Newton's Method).
\item[(2)] Store the solution as $Z^{(k)}$. Update $w^{(k)}=\mu w^{(k-1)}$. 
\item[(3)] Stop when $w^{(k)}= \dfrac{n}{\epsilon}$.
\end{itemize}

The number of outer iterations for a given error tolerance is constant over $n$, but can vary with the factor $\mu$ by which the weights $w^{(k-1)}$ vary over the steps. \footnote{See section 11.5.3 of \cite{Boyd} for an upper bound on the number of outer iterations; particularly note figure 11.14. Also see the second example of section 11.6.3., figure 11.16 reveals the variation of the number of outer iterations with $\mu$.}. Thus the computational complexity of the algorithm is decided by the computational complexity of Newton's method within the inner iterations. In the following table we list the different steps as part of Newton's algorithm and give the computational complexity (time and space) for each step. \medbreak

\small
\begin{tabular}{ | p{1em} |p{25em}|p{7em}| p{7em} |}
 \hline
 \multicolumn{4}{|c|}{Computational Complexity for Newton's Method} \\
 \hline
 & Step in the algorithm     & \small{Time Complexity} & \small{Space Complexity} \\
 \hline
 1. & Computing values of $h_{ij}(\vv{y}^{(k-1)})$, by substituting $\vv{y}^{(k-1)}$ into equations \eqref{hij}, for all $ 1\le i,j \le n$.  & $\mathcal{O}(n^2)$    & $\mathcal{O}(n^2)$\\
 
 2. & Computing the Jacobian $H$, at the point $\vv{y}^{(k-1)}$  & $\mathcal{O}(n^5)$    & $\mathcal{O}(n^4)$\\
 
 3. & Computing the inverse of $H$ at the point $\vv{y}^{(k-1)}$  &   $\mathcal{O}(n^6)$  & $\mathcal{O}(n^4)$\\
 
 4. & Computing $y_{ij}^{(k)} = y_{ij}^{(k-1)} -\sum_{s,t=1}^{n} \left(H^{-1}\right)_{ij,st}h_{st}^{(k-1)}$, $\forall \; i,j, \le n$ & $\mathcal{O}(n^4)$ & $\mathcal{O}(n^2)$ \\
 
  \hline
\end{tabular}\normalsize \medbreak

\textbf{Comparing Different Methods:} The table above shows that the computational complexity of the Barrier-type IPM is as costly as the direct application of Newton's method. In fact, a closer analysis shows that directly applying Newton's method is less costly than the SDP method, along with the advantage of being simpler to implement. Also, the Taylor series method is as costly as both Newton's method and the SDP method, when one can find a gram matrix $G_0$ in the close vicinity of the given gram matrix $G_1$. If one is interested in a one-time calculation for an ensemble of LI pure states Newton's method is the most desirable method to implement among all the three examined here.

\section{Remarks and Conclusion}
\label{conclusion}

We showed how the mathematical structure of the MED problem for LI pure state ensembles could be used to obtain the solution for said problem. This was done by casting the necessary and sufficient conditions \eqref{St} and \eqref{Glb} into a rotationally invariant form which was employed to obtain the solution by using the implicit function theorem. We also showed that this technique is simpler to implement than standard SDP techniques. 

As mentioned in section \eqref{intro}, for fixed states $\ket{\psi_1},\ket{\psi_2},\cdots,\ket{\psi_n}$, $\mathscr{R}$ induces an invertible map on the space of probabilities, $\{p_i\}_{i=1}^n \longrightarrow \{q_i\}_{i=1}^{n}$. This naturally begs a question on whether there is a relation between the two probabilities, for example does one majorize the other? Or, more generally, is the entropy of $\{q_i\}_{i=1}^{n}$ always larger than the entropy of $\{p_i\}_{i=1}^{n}$ or vice versa? The answer to this question is that there doesn't seem to be any simple property relating these two probabilities, vis-a-vis, one majorizing the other or that the $(\{p_i\}_{i=1}^{n},\{q_i\}_{i=1}^{n})$-pair are either related by $H(p_i) \geq H(q_i)$ or $H(p_i) < H(q_i)$; examples of both cases can be found. 

In this paper we studied only about the case for $n$-LIP ensembles. Naturally there is the question if a similar theory holds for more general ensembles. For the case of $m$-linearly dependent pure state ensembles (where $m>dim \mathcal{H}=n$): it is explicitly shown that, while a map like $\mathscr{R}^{-1}$ exists on the space of $m$ linearly dependent pure (LDP) state ensembles, $\mathscr{R}^{-1}$ isn't one-to-one \cite{Carlos}\footnote{In that sense it defeats the purpose of denoting such a map by $\mathscr{R}^{-1}$, because $\mathscr{R}^{-1}$ doesn't have an \emph{inverse}.}. From the analysis in our paper, it is clear that the one-to-one nature of the map $\mathscr{R}^{-1}$ (for $n$-LIPs) plays a crucial in formulating the rotationally invariant necessary and sufficient conditions for MED of said ensemble of states; thus it also plays a crucial role in the application of this necessary and sufficient condition to obtain the solution for MED of said ensemble. The non-invertibility of $\mathscr{R}^{-1}$ 
also shows that the optimal POVM won't necessarily vary smoothly as one varies the ensemble from one $m$-LDP to another $m$-LDP. C. Mochon gave algebraic arguments for this \cite{Carlos} in his paper, and Ha et al. showed the same using the geometrical arguments for ensembles of three qubit states, as an example \cite{Ha}. This has been shown for general qudits as well \cite{Kwon}. Besides this, there is also the fact that there are some LDPs for which the optimal POVM isn't even unique, i.e., two or more distinct POVMs give the maximum success probability for MED. This means that as the ensemble is varied in the neighbourhood of said ensemble, the optimal POVM can undergo discontinuous \emph{jumps}. Hence, we conclude that such the technique which was used in section \ref{empIFT} for $n$-LIPs can't be generalized to $m$-LDPs. Work is currently under progress to see if such a technique can be generalized to mixed state ensembles of linearly independent states \cite{SGMix14}.

\paragraph{Acknowledgements}
The authors wish to thank Dr. R. Simon and Jon Tyson for meaningful discussions and insightful remarks. 

\bibliographystyle{hplain}
\bibliography{MED}

\end{document}